%% file: main.tex
\documentclass[pra,twocolumn,superscriptaddress,amsfonts,amssymb,amsmath,nofootinbib]{revtex4-2}
\pdfoutput=1
\newcommand{\externalizetikz}{0}   
\newcommand{\dowordcount}{0}

\usepackage{hyperref}
\hypersetup{colorlinks,urlcolor=blue,linkcolor=black,citecolor=blue}
\usepackage[capitalize,nameinlink]{cleveref}
\usepackage[normalem]{ulem}
\usepackage[utf8]{inputenc}
\usepackage{amsthm}
\usepackage{mathtools,mymath}
\newtheorem{theorem}{Theorem}
\newtheorem{definition}{Definition}

\DeclareMathOperator{\swap}{SWAP}
\DeclareMathOperator{\vecc}{vec}

\usepackage{xcolor}
\usepackage[caption=false]{subfig}

\usepackage{braket}
\usepackage{qcircuit}
\usepackage{environ}
\newcommand{\myqctmp}[2][0.25]{\Qcircuit @C=#2em @R=#1em @!R}
\NewEnviron{myqcircuit}[1][0.25]{\vcenter{\myqctmp[#1]{0.75} {\BODY}}}
\NewEnviron{myqcircuit*}[2]{\vcenter{\myqctmp[#1]{#2} {\BODY}}}
\newcommand{\controlsq}{*!<0em,.025em>-=-<0em>{\square}} 
\newcommand{\ctrlsq}[1]{\controlsq \qwx[#1] \qw}         

\usepackage{environ}
\ifnum\dowordcount=1
  \NewEnviron{myequation}{}
  \NewEnviron{myequation*}{}
  \NewEnviron{myalign}{}
  \NewEnviron{myalign*}{}
\else
  \NewEnviron{myequation}{\begin{equation}\BODY\end{equation}}
  \NewEnviron{myequation*}{\begin{equation*}\BODY\end{equation*}}
  \NewEnviron{myalign}{\begin{align}\BODY\end{align}}
  \NewEnviron{myalign*}{\begin{align*}\BODY\end{align*}}
\fi

\definecolor{myblue}{rgb}{0,0.4470,0.7410}
\definecolor{myred}{rgb}{0.8500,0.3250,0.0980}
\definecolor{myorange}{rgb}{0.9290,0.6940,0.1250}
\definecolor{mypurple}{rgb}{0.4940,0.1840,0.5560}
\definecolor{mygreen}{rgb}{0.4660,0.6740,0.1880}
\definecolor{mylightblue}{rgb}{0.3010,0.7450,0.9330}
\definecolor{mydarkred}{rgb}{0.6350,0.0780,0.1840}

\ifnum\externalizetikz=2
  \usepackage{tikzexternal}
  \tikzsetexternalprefix{fig/}
\else
  \usepackage{tikz}
  \usepackage{pgfplots}
  \usetikzlibrary{calc}
  \usetikzlibrary{positioning}
  \usetikzlibrary{plotmarks}
  \pgfplotsset{
    compat=newest,
    table/header=false,
    tick label style={font=\footnotesize},
    label style={font=\small},
    legend style={font=\footnotesize},
    legend cell align=left,
    colormap={parula}{
      rgb255=(53,42,135)
      rgb255=(15,92,221)
      rgb255=(18,125,216)
      rgb255=(7,156,207)
      rgb255=(21,177,180)
      rgb255=(89,189,140)
      rgb255=(165,190,107)
      rgb255=(225,185,82)
      rgb255=(252,206,46)
      rgb255=(249,251,14)
    }
  }
  \pgfplotsset{
    myColOne/.style={myblue},
    myColTwo/.style={myred},
    myColThr/.style={myorange},
    myColFou/.style={mypurple},
    myColFiv/.style={mygreen},
    myColSix/.style={mylightblue},
    myColSev/.style={mydarkred}
  }
  \ifnum\externalizetikz=1
    \usepgfplotslibrary{external}
  \fi
\fi

\usepackage{tabularx,multirow,booktabs}
\newcolumntype{L}{>{\raggedright\arraybackslash}X}
\newcolumntype{C}{>{\centering\arraybackslash}X}
\newcolumntype{R}{>{\raggedleft\arraybackslash}X}

\begin{document}

\title{FABLE: Fast Approximate Quantum Circuits for Block-Encodings}

\author{Daan Camps}
\email{dcamps@lbl.gov}
\affiliation{National Energy Research Scientific Computing Center,
            Lawrence Berkeley National Laboratory,
            Berkeley, CA 94720, USA}
            
\author{Roel~Van~Beeumen}
\email{rvanbeeumen@lbl.gov}
\affiliation{Applied Mathematics and Computational Research Division,
            Lawrence Berkeley National Laboratory,
            Berkeley, CA 94720, USA}

\date{\today}

\begin{abstract}
Block-encodings of matrices have become an essential element of quantum algorithms derived from the
quantum singular value transformation. 
This includes a variety of algorithms ranging from the quantum linear systems
problem to quantum walk, Hamiltonian simulation, and quantum machine learning.
Many of these algorithms achieve optimal complexity in terms of black box matrix oracle queries, but so far the
problem of computing quantum circuit implementations for block-encodings of matrices has been
under-appreciated. In this paper we propose FABLE, a method to generate approximate quantum circuits for
block-encodings of matrices in a fast manner.
FABLE circuits have a simple structure and are directly formulated in terms of one- and two-qubit gates. 
For small and structured matrices they are feasible in the NISQ era, and the circuit parameters can
be easily generated for problems up to fifteen qubits.
Furthermore, we show that FABLE circuits can be compressed and sparsified. 
We provide a compression theorem that
relates the compression threshold to the error on the block-encoding.
We benchmark our method for Heisenberg and Hubbard Hamiltonians, and Laplacian operators to 
illustrate that they can be implemented with a reduced gate complexity without approximation error.
\end{abstract}

\maketitle

\section{Introduction}
The quantum singular value transformation (QSVT) \cite{gisu2018, gisu2019} combines
and extends qubitization \cite{loch2019} and quantum signal processing \cite{loch2017}.
The QSVT provides a unifying framework encompassing many quantum algorithms that
provide a speed-up over the best known classical algorithm~\cite{Martyn2021}.
All quantum algorithms derived from the quantum singular value transformation ultimately 
rely on the notion of a \textit{block-encoding} of some matrix $A$ that represents
the problem at hand. This matrix can, for example, be the Hamiltonian of
the quantum system to be simulated~\cite{Berry2015},
the discriminant matrix of the Markov
chain in a quantum walk~\cite{szeg2004, quant-ph/0401053},
or the matrix to be solved for in the quantum
linear systems problem~\cite{Childs2017}.
These matrices are, in general, non-unitary operators and cannot be directly run
on a quantum computer that only performs unitary evolution.
This constraint is usually overcome by enlarging the Hilbert space and
embedding the non-unitary operator in a specific state of the ancillary
qubits.
The most commonly used embedding is a block-encoding where the system matrix
$A$ is embedded in the leading principal block of a larger unitary matrix acting
on the full Hilbert space:
\begin{myequation}
U = \begin{bmatrix} A & *\, \\ * & *\, \end{bmatrix},
\label{eq:BE}
\end{myequation}
where $*$ indicate arbitrary matrix elements.
We assume that $\|A\|_{2} \leq 1$
since otherwise such an embedding cannot exist.
In this case, we say that $U$ block encodes $A$.

Thus far, complexity results for quantum algorithms derived from the quantum singular
value transformation have been formulated in terms of query complexity, i.e.,
how many queries to the block-encoding $U$ are required to solve the problem.
The question how a quantum circuit can be generated for \cref{eq:BE} 
has so far been under appreciated.
Encoding schemes for sparse matrices have been proposed~\cite{Childs2017,gisu2018} and
ultimately rely on sparse access oracles for which a quantum circuit implementation can
be challenging.
Some explicit circuit implementations for quantum walks on highly-structured
graphs are provided in~\cite{Loke2017} and are closely connected to block-encodings~\cite{gisu2018}.
Similarly, \cite{yang22} shows how to generate quantum circuits for block-encoding
certain specific $2^n \times 2^n$ sparse matrices in poly$(n)$ complexity.

In this paper we take a more general approach and propose \emph{FABLE}, which
stands for Fast Approximate BLock-Encodings. FABLE is 
an efficient algorithm to generate quantum circuits that block encode arbitrary matrices 
up to prescribed accuracy.
FABLE circuits consist of a \emph{matrix query oracle}, that we implement
with  simple one-qubit $R_y$ and $R_z$ rotations and two-qubit CNOT gates,
and some additional Hadamard and SWAP gates.
The gate complexity of a FABLE circuit for general, unstructured $N \times N$
matrices is bounded by $\bigO(N^2)$ gates with a modest prefactor of 2 
for real-valued matrices (4 for complex-valued matrices), and a 
limited polylogarithmic overhead. 
In this sense, FABLE circuits are a quantum circuit representation of dense matrices 
with an optimal asymptotic gate complexity.
However, this gate complexity scales exponentially in the number of qubits for
generic dense matrices as encoding an unstructured matrix of exponential
dimension is a difficult task.
Luckily, more relevant problems usually contain a lot of structure
and, as we will show, FABLE circuits can be compressed which often leads
to significantly reduced gate complexities for many problems of interest.
This process can be interpreted as a \emph{sparsification} of the circuit and matrix.
However, the FABLE algorithm applies a Walsh-Hadamard transformation to the matrix data
and performs the sparsification in the Walsh-Hadamard domain.
Sparse FABLE matrices thus corresponds to matrices that contain many zeros in this domain,
which does not necessarily correspond to sparsity in the original domain.
Thanks to these characteristics, FABLE circuits are well-suited for implementing
quantum algorithms derived from the QSVT in the NISQ era and beyond.

The remainder of this paper is organized as follows.
\Cref{sec:BE} formally defines the concept of block-encodings. \Cref{sec:circ}
presents a circuit structure for block-encodings in terms of a matrix query
oracle and gives a naive implementation of this oracle.
\Cref{sec:FABLE} introduces the improved circuit implementation for matrix
query oracles that are used in FABLE circuits for real- and complex-valued
matrix data.
In \Cref{sec:approx}, we extended the definition of block-encodings to
approximate block-encodings and we discuss how FABLE circuits can be compressed.
We relate the threshold in the compression algorithm to the error on the block-encoding.
\Cref{sec:ex} provides some examples that show that FABLE can be used
to block encode Heisenberg and Hubbard Hamiltonians and discretized differential
operators with a significantly reduced gate complexity.

Implementations of the FABLE algorithm built on top of QCLAB~\cite{qclab, qclabpp}
and Qiskit~\cite{Qiskit} are made publicly available on \url{https://github.com/QuantumComputingLab/fable}.
Without loss of generality, we assume in the remainder of this paper that
the matrix size is $N \times N$ with $N = 2^n$.

\section{Block-encodings}\label{sec:BE}

A \emph{block-encoding} of a non-unitary matrix is the embedding of a properly scaled
version of that matrix in the leading principal block of a bigger unitary
\cite{gisu2018,gisu2019}.
A formal definition for a block-encoding of an $n$-qubit matrix $A$
in an $m$-qubit unitary $U$ is as follows.

\begin{definition}
\label{def:BE}
Let $a,n, m \inN$, $m = a + n$.
Then an $m$-qubit unitary $U$ is a ${(\alpha, a)}$-block-encoding 
of an $n$-qubit operator $A$ if
\begin{myequation}
\tilde A = \left(\bra{0}^{\otimes a} \otimes \eI_{n} \right) U
               \left(\ket{0}^{\otimes a} \otimes \eI_{n} \right).
\end{myequation}
and $A = \alpha \tilde A$.
\end{definition}

The parameters $(\alpha, a)$ are, respectively, 
the \emph{subnormalization factor} necessary for encoding matrices of arbitrary norm
and the number of \emph{ancilla} qubits used in the block-encoding.
Since $\normtwo{U} = 1$, we have that $\normtwo{\tilde A} \leq 1$ and therefore $\normtwo{A} \leq \alpha$.
Every unitary is already a trivial ${(1,0)}$-block-encoding of itself 
and every non-unitary matrix can be embedded in a ${(\normtwo{A},1)}$-block-encoding~\cite{QI:Alber:2001}.
This does not guarantee the existence of an efficient quantum circuit implementation, but merely
considers the matrix representation.

For a block-encoding, we say that $\tilde A$ is the partial trace of $U$ over the zero state of the ancilla space. 
This naturally partitions the Hilbert space $\cH_m$ into $\cH_a \otimes \cH_n$.
Given an $n$ qubit state, $\ket{\psi} \in \cH_n$,
the action of $U$ on $\ket{\phi} = \ket{0}^{\otimes a} \otimes \ket{\psi}$ becomes
\begin{myequation}
U \ket{\phi} = \ket{0}^{\otimes a} \otimes \tilde A \ket{\psi}
+ \sqrt{1 - \norm{ \tilde A \ket{\psi}}^2} \ket{\sigma^{\perp}},
\end{myequation}
with
\begin{align}
\left(\bra{0}^{\otimes a} \otimes \eI_{n}\right) \ket{\sigma^\perp} &= 0, \\
\norm[\big]{\ket{\sigma^\perp}} &= 1,
\end{align}
and $\ket{\sigma^\perp}$ the normalized state for which the ancilla register 
has a state orthogonal to $\ket{0}^{\otimes a}$.
With probability $\norm{\tilde A \ket{\psi}}^2$, a partial measurement
of the ancilla qubits results in $0^{\otimes a}$ and the signal qubits are
in the target state $\tilde A \ket{\psi}/\norm{\tilde A \ket{\psi}}$.
Using amplitude amplification~\cite{grover, gisu2019}, this process must be repeated 
$1/\norm{\tilde A \ket{\psi}}$ times for success on average.

\begin{figure}[hbtp]
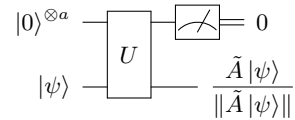

\centering\ifnum\dowordcount=0
\[\begin{myqcircuit*}{1.25}{1}
\lstick{\ket{0}^{\otimes a}} & \multigate{1}{U} & \meter & \rstick{0}\cw \\
\lstick{\ket{\psi}}        & \ghost{U}        & \rstick{\Frac{\tilde A\ket\psi}{\norm{\tilde A\ket\psi}}}\qw
\end{myqcircuit*}\]
\fi
\caption{Abstract quantum circuit for an ${(\alpha, a)}$ block-encoding $U$ of $A$.
The lower quantum wire carries the \emph{signal} qubits, the upper wire are the \emph{ancilla} qubits. 
If the ancilla register is measured in the zero state, the signal register 
is in the desired 
state ${\tilde A\ket\psi}/{\norm{\tilde A\ket\psi}}$.
\label{fig:BE}}
\end{figure}

\cref{fig:BE} provides the high-level structure of a quantum circuit for a block-encoding.
We note that an encoding of a matrix can be coupled to any other state of the
ancilla space besides the all-zero state and this state does not even have to
be a computational basis state. This generalization is discussed in~\cite{gisu2018}.
Encoding the matrix data in the all-zero state of the
ancilla qubits has become the most widely used choice as this leads to an embedding
of the matrix in the leading block of the unitary matrix, i.e., a block-encoding.

\section{Quantum Circuits for Block-Encodings}
\label{sec:circ}

Quantum circuits for block-encodings of sparse matrices are often presented in
terms of query oracles that provide information about the position
and binary description of the matrix entries~\cite{gisu2019,Childs2017}.
These oracles can be combined into a \emph{matrix query oracle}~\cite{lin2022}
that provides access to the matrix data.

Our approach is different, 
we immediately define a matrix query operation $O_A$ for a given matrix $A$
and consequently discuss how
this oracle can be directly synthesized in a quantum circuit.

\begin{definition}
\label{def:OA}
Let $a_{ij}$ be the elements of an $N \times N$ matrix $A$ with $N = 2^n$,
$\|a_{ij}\| \leq 1$. 
Then the matrix query operation $O_A$ applies
\begin{equation}
O_A \ket0 \ket{i} \ket{j} =
  \left( a_{ij} \ket0 + \sqrt{1 - |a_{ij}|^2} \ket1 \right) \ket{i} \ket{j},
\label{eq:OA}
\end{equation}
where $\ket{i}$ and $\ket{j}$ are $n$-qubit computational basis states.
\end{definition}

A high-level quantum circuit to block-encode a matrix $A$ is proposed in
\cref{fig:BEcirc}, where $H^{\otimes n}$ is an $n$-qubit Hadamard transformation
that creates an equal superposition over the row qubits,
the matrix query unitary $O_A$ is given in \eqref{eq:OA}, and the $2n$-qubit $\swap$
gate is implemented as $\swap\ket{i}\ket{j} = \ket{j}\ket{i}$.
This circuit is closely related to similar circuits in~\cite{lin2022},
but we encode all information about the matrix in a single matrix query oracle.

\begin{figure}[hbtp]
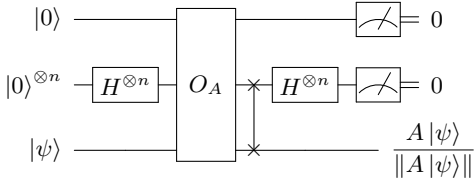

\centering\ifnum\dowordcount=0
\[
\begin{myqcircuit}[1.25]
\lstick{\ket{0}}             & \qw                  & \multigate{2}{O_A} &
  \qw                        & \qw                  & \meter & \rstick{0}\cw \\
\lstick{\ket{0}^{\otimes n}} & \gate{H^{\otimes n}} &        \ghost{O_A} &
  \qswap                     & \gate{H^{\otimes n}} & \meter & \rstick{0}\cw \\
\lstick{\ket{\psi}}          & \qw                  &        \ghost{O_A} &
  \qswap\qwx                 & \qw                  &
  \rstick{\Frac{A\ket\psi}{\norm{A\ket\psi}}}\qw
\end{myqcircuit}
\]
\fi
\caption{High-level quantum circuit structure for block-encoding a matrix $A$ 
in terms of a matrix query oracle $O_A$.
If the $n+1$ ancilla qubits are measured in the zero state, the signal register 
is in the desired 
state ${\tilde A\ket\psi}/{\norm{\tilde A\ket\psi}}$.
\label{fig:BEcirc}}
\end{figure}

The following theorem 
ascertains that the circuit from \cref{fig:BEcirc}
is indeed a block-encoding of the target matrix $A$.
Our proof follows a similar reasoning to~\cite{lin2022} and is 
included for completeness.

\begin{theorem}
\label{thm:UA}
The circuit $U_A$ in \cref{fig:BEcirc} is an $(1/2^n, n+1)$-block-encoding
of the $n$-qubit matrix $A$ if $|a_{ij}| \leq 1$.
\end{theorem}
\begin{proof}
The circuit $U_A$ can be written in matrix notation as:
\begin{equation*}
    U_A = (\eye[1] \otimes H^{\otimes n} \otimes \eye[n])
          (\eye[1] \otimes \swap) O_A
          (\eye[1] \otimes H^{\otimes n} \otimes \eye[n]).
\end{equation*}
For $U_A$ to be an $(1/2^n, n+1)$-block-encoding of $A$, we need
to verify according to \Cref{def:BE} that
\begin{equation*}
    \bra{0} \bra{0}^{\otimes n} \bra{i} \, U_A \, \ket{0} \ket{0}^{\otimes n} \ket{j} =
    \frac{1}{2^n} a_{ij}.
\end{equation*}
On one hand, we have
\begin{align*}
    \ket{0} \ket{0}^{\otimes n} \ket{j} & \\
    \quad        \xrightarrow{H^{\otimes n}} & \ \frac{1}{\sqrt{2^n}} \sum_{k=0}^{2^n-1} \ket{0} \ket{k} \ket{j}, \\
             \xrightarrow{O_A} & \ \frac{1}{\sqrt{2^n}} \sum_{k=0}^{2^n-1} \left(a_{kj} \ket{0} + \sqrt{1 - |a_{kj}|^2} \ket{1} \right) \ket{k} \ket{j}, \\
             \xrightarrow{\swap} & \ \frac{1}{\sqrt{2^n}} \sum_{k=0}^{2^n-1} \left(a_{kj} \ket{0} + \sqrt{1 - |a_{kj}|^2} \ket{1} \right) \ket{j} \ket{k},
\end{align*}
while on the other hand, we have
\begin{equation*}
    \ket{0} \ket{0}^{\otimes n} \ket{i}  \xrightarrow{H^{\otimes n}} \ \frac{1}{\sqrt{2^n}} \sum_{\ell=0}^{2^n-1} \ket{0} \ket{\ell} \ket{i}.
\end{equation*}
Combining both, yields
\begin{align*}
    \bra{0} \bra{0}^{\otimes n} & \bra{i} \, U_A \, \ket{0} \ket{0}^{\otimes n} \ket{j} \\
    = & \ \frac{1}{2^n} 
    \left( \sum_{\ell=0}^{2^n-1} \bra{0} \bra{\ell} \bra{i} \right) \\
    & \ \quad \ \left( \sum_{k=0}^{2^n-1} \left(a_{kj} \ket{0} + \sqrt{1 - |a_{kj}|^2} \ket{1} \right) \ket{j} \ket{k} \right), \\
    = & \ \frac{1}{2^n} \sum_{\ell=0}^{2^n-1} \sum_{k=0}^{2^n-1} a_{kj} \braket{\ell | j} \braket{i | k} \\
    = & \ \frac{1}{2^n} a_{ij}.
\end{align*}
which completes the proof.
\end{proof}

All circuit elements in \cref{fig:BEcirc}, except for $O_A$, can be readily 
written in simple 1- and 2-qubit gates. The complexity of the Hadamard and SWAP
gates is only poly$(n)$.
We present next how the oracle $O_A$ from
\Cref{def:OA} can be implemented in simple 1- and 2-qubit gates 
for arbitrary matrices. We first consider the cases of real-valued matrices
and discuss complex-valued matrices in~\Cref{sec:complex}.

In case that $A$ is a real-valued matrix, we see from \eqref{eq:OA}
that for given row and column indices $i$ and $j$, $O_A$ acts on the
$\ket0$ state of the first qubit as an $R_y$ gate with angle
\begin{equation}
\theta_{ij} = \arccos(a_{ij}),
\label{eq:theta}
\end{equation}
i.e.,
\begin{align}
R_y(2\theta_{ij}) \ket0 &=
\begin{bmatrix}
\cos(\theta_{ij}) &           -\sin(\theta_{ij}) \\
\sin(\theta_{ij}) & \phantom{-}\cos(\theta_{ij})
\end{bmatrix}
\begin{bmatrix} 1 \\ 0 \end{bmatrix}, \nonumber \\
 &=
\begin{bmatrix}
a_{ij} \\ \sqrt{\smash[b]{1 - a_{ij}^2}}
\end{bmatrix}.
\end{align}
Hence, the matrix query unitary $O_A$ is a matrix with the following structure
for real-valued matrices
\begin{equation}
O_A = 
\left[\begin{smallmatrix}
c_{00} & & & & -s_{00} \\
 & c_{01} & & & & -s_{01} \\
 & & \ddots & & & & \ddots \\
 & & & c_{N-1,N-1} & & & & -s_{N-1,N-1} \\
s_{00} & & & & c_{00} \\
 & s_{01} & & & & c_{01} \\
 & & \ddots & & & & \ddots \\
 & & & s_{N-1,N-1} & & & & c_{N-1,N-1}
\end{smallmatrix}\right],
\label{eq:OAmat}
\end{equation}
where $c_{ij} := \cos(\theta_{ij})$ and $s_{ij} := \sin(\theta_{ij})$,
with $\theta_{ij}$ given by \eqref{eq:theta}.

A first naive implementation of the $O_A$ oracle \eqref{eq:OA} 
for $A\inR[N][N]$ uses $N^2$ multi-controlled $R_y$ gates.
We use the notation $C^n(R_y)$ for an $R_y$ gate with $n$ control qubits.
The circuit construction for $O_A$ uses one $C_n(R_y)$ gate for each
matrix entry $a_{ij}$ where the control qubits encode the row and column 
indices $\ket{i}\ket{j}$ of the corresponding entry.
This circuit is illustrated below for the encoding of a 
$2 \times 2$ matrix using 3 qubits and 4 $C^2(R_y)$ gates
\[
\begin{myqcircuit}
 & \gate{R_y(2\theta_{00})} & \gate{R_y(2\theta_{01})} &
   \gate{R_y(2\theta_{10})} & \gate{R_y(2\theta_{11})} & \qw \\
 & \ctrlo{-1} & \ctrlo{-1} & \ctrl{-1}  & \ctrl{-1} & \qw \\
 & \ctrlo{-1} & \ctrl{-1}  & \ctrlo{-1} & \ctrl{-1} & \qw
\end{myqcircuit}
\]
where $\theta_{00}, \theta_{01}, \theta_{10}, \theta_{11}$ are given by
\eqref{eq:theta}.
It is clear that this circuit implements \eqref{eq:OA} for real-valued
data.
The major disadvantage of this naive approach is that it requires 
$N^2$ $C^{2n}(R_y)$ gates to implement the $O_A$ oracle for $A \inR[N][N]$.
However, every $C^{2n}(R_y)$ requires $\bigO(N^2)$
1- and 2-qubit gates to be implemented~\cite{Barenco1995}.
This brings the total gate complexity of this naive circuit implementation
to $\bigO(N^4)$ which is excessive as it is quadratically worse than the 
classical representation cost.
We propose a quadratic reduction in gate complexity with the
FABLE implementation of block-encoding circuits in the next section.

\section{FABLE Circuits for Block-Encodings}\label{sec:FABLE}

We continue with the case of real-valued matrices and discuss how to 
extend this to complex-valued matrices afterwards.

\subsection{Query oracles for real-valued matrices}

We illustrate the idea of the improved circuit construction for $O_A$
for a small-scale example of $A \inR[2][2]$ from the previous section.
In that case, the circuit structure is given by
\[
\begin{myqcircuit}
	& \gate{R_y(\hat \theta_0)} & \targ     & \gate{R_y(\hat \theta_1)} & \targ     & \gate{R_y(\hat \theta_2)} & \targ     & \gate{R_y(\hat \theta_3)} & \targ     & \qw \\
	& \qw                       & \qw       & \qw                       & \ctrl{-1} & \qw                       & \qw       & \qw                       & \ctrl{-1} & \qw \\
	& \qw                       & \ctrl{-2} & \qw                       & \qw       & \qw                       & \ctrl{-2} & \qw                       & \qw       & \qw
\end{myqcircuit}
\]
where the angles $\hat\theta_0, \hat\theta_1, \hat\theta_2, \hat\theta_3$
are computed from the data $A\inR[2][2]$ as we will explain next.
The circuit structure is derived from~\cite{Mottonen2004}.

We analyze the action of the above circuit based on the following two elementary
properties of $R_y$ rotations:
\begin{equation}
\begin{split}
R_y(\theta_0) \, R_y(\theta_1) & = R_y(\theta_0 + \theta_1), \\
X \, R_y(\theta) \, X & = R_y(-\theta),
\end{split}\label{eq:rotcond}
\end{equation}
It follows that the state of the first qubit is rotated as
\begin{align*}
00&: & \phantom{X} R_y(\hat \theta_3) \phantom{X} R_y(\hat \theta_2) \phantom{X}  R_y(\hat \theta_1) \phantom{X}  R_y(\hat \theta_0) = \qquad\qquad\qquad& \\ && R_y( \phantom{-}\hat\theta_3 + \hat\theta_2 + \hat\theta_1 + \hat\theta_0), \\
01&: & \phantom{X} R_y(\hat \theta_3) X R_y(\hat \theta_2) \phantom{X}  R_y(\hat \theta_1) X R_y(\hat \theta_0) = \qquad\qquad\qquad& \\ && R_y( \phantom{-}\hat\theta_3 - \hat\theta_2 - \hat\theta_1 + \hat\theta_0),\\
10&: & X R_y(\hat \theta_3) \phantom{X} R_y(\hat \theta_2) X R_y(\hat \theta_1) \phantom{X}  R_y(\hat \theta_0) = \qquad\qquad\qquad& \\ && R_y( - \hat\theta_3 - \hat\theta_2 + \hat\theta_1 + \hat\theta_0),\\
11&: & X R_y(\hat \theta_3) X R_y(\hat \theta_2) X R_y(\hat \theta_1) X R_y(\hat \theta_0) = \qquad\qquad\qquad& \\ && R_y( - \hat\theta_3 + \hat\theta_2 - \hat\theta_1 + \hat\theta_0),
\end{align*}
where the rotation angle depends on the state of the last two control qubits as indicated above.
To implement an $O_A$ oracle with angles $\theta_{00}$, $\theta_{01}$, $\theta_{10}$, $\theta_{11}$ 
as given by \eqref{eq:theta}, we \emph{vectorize} $A$ to $\vecc(A)$ in row-major order
such that $\vecc(A)_{i + j \cdot N} = a_{ij}$ to obtain relabeled angles
$(\theta_0, \ldots, \theta_3)$. We see from the system of equation above 
that these angles have to satisfy
\begin{equation}
\begin{bmatrix}
\theta_0 \\ \theta_1 \\ \theta_2 \\ \theta_3
\end{bmatrix}
=
\begin{bmatrix}
1 & \phantom{-}1 & \phantom{-}1 & \phantom{-}1 \\
1 & -1           & -1           & \phantom{-}1 \\
1 & \phantom{-}1 & -1           & -1           \\
1 & -1           & \phantom{-}1 & -1
\end{bmatrix}
\begin{bmatrix}
\hat\theta_0 \\ \hat\theta_1 \\ \hat\theta_2 \\ \hat\theta_3
\end{bmatrix}.
\end{equation}
This is a structured linear system that can be written as
\begin{align}
\begin{bmatrix}
\theta_0 \\ \theta_1 \\ \theta_2 \\ \theta_3
\end{bmatrix}
& =
\begin{bmatrix}
1 & \phantom{-}1 & \phantom{-}1 & \phantom{-}1 \\
1 & -1           & \phantom{-}1           & -1 \\
1 & \phantom{-}1 & -1           & -1           \\
1 & -1           & -1 & \phantom{-}1
\end{bmatrix}
\begin{bmatrix}
1 &   &   &   \\
  & 1 &   &   \\
  &   & 0 & 1 \\
  &   & 1 & 0
\end{bmatrix}
\begin{bmatrix}
\hat\theta_0 \\ \hat\theta_1 \\ \hat\theta_2 \\ \hat\theta_3
\end{bmatrix}, \\
& = 
( \hat H \otimes \hat H ) P_G
\begin{bmatrix}
\hat\theta_0 \\ \hat\theta_1 \\ \hat\theta_2 \\ \hat\theta_3
\end{bmatrix},
\label{eq:sys}
\end{align}
where $\hat H = \left[\begin{smallmatrix} 1 & \phantom{-}1 \\ 1 & -1 \end{smallmatrix}\right]$ is 
a scalar multiple of the Hadamard gate and $P_G$ is the permutation matrix that transforms binary ordering 
to Gray code ordering.

This algorithm generalizes to $O_A$ oracles for matrices $A \inR[N][N]$~\cite{Mottonen2004}.
The corresponding circuit structure consists of a gate sequence of length $2^{2n}$ alternating between
$R_y$ and CNOT gates.
Note that the $R_y$ gates only act on the first qubit,
which is also the target qubit of the CNOT gates,
and the control qubit for the
$\ell$th CNOT gate is determined by the bit where the $\ell$th and $(\ell + 1)$st Gray code differ.
For an $O_A$ oracle with angles $\bftheta = (\theta_0, \ldots, \theta_{2^{2n}-1})$ given by \eqref{eq:theta},
the angles of the $R_y$ gates in the quantum circuit, $\hat \bftheta = (\hat \theta_0, \ldots, \hat \theta_{2^{2n}-1})$, are related to $\bftheta$ through the linear system
\begin{align}
\left( \hat H^{\otimes 2n} \, P_G \right) \hat \bftheta = \bftheta.
\label{eq:ls}
\end{align}
This linear system can be efficiently solved by a classical algorithm in
$\bigO(N^2\log N^2)$ using a fast Walsh--Hadamard transform~\cite{fwht}
which is implemented in the reference implementation of FABLE.

This circuit structure is known as a \emph{uniformaly
controlled $R_y$ rotation}~\cite{Mottonen2004} because it rotates the
target qubit over a different angle for each bitstring in the control
register. We use the following concise notation for uniformly controlled
rotations used in the implementation of $O_A$: 
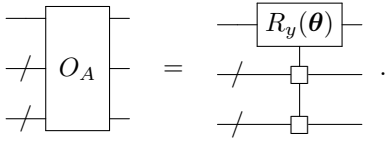
\[
\begin{myqcircuit*}{1}{0.75}
\lstick{} & \qw & \multigate{2}{O_A} &  \qw \\
\lstick{} & {/} \qw &        \ghost{O_A} &  \qw \\
\lstick{} & {/} \qw &        \ghost{O_A} &  \qw
\end{myqcircuit*}
\quad = \quad
\begin{myqcircuit}
& \qw & \gate{R_y(\bftheta)} & \qw \\
& {/} \qw & \ctrlsq{-1} & \qw \\
& {/} \qw & \ctrlsq{-1} & \qw
\end{myqcircuit}\ \ .
\]
The gate complexity of implementing $O_A$ with this approach
is $\bigO(N^2)$ for $A \inR[N][N]$, where $N^2$ CNOT and $N^2$
single-qubit $R_y$ gates are required, i.e., the prefactor in
the circuit complexity is 2.
This reaches the same asymptotic gate complexity as classically
required to store the data and is optimal for unstructured data.

\subsection{Query oracles for complex-valued matrices}
\label{sec:complex}
In case that $A$ is a complex-valued matrix, we encode 
the matrix elements as a product of $R_y$ and $R_z$ rotations.
For given row and column indices $i$ and $j$, 
the matrix element to be encoded is 
$a_{ij} = |a_{ij}| e^{\I \alpha_{ij}}$, and 
$O_A$ acts on the $\ket0$ state of the first qubit
as a product of an $R_y$ and $R_z$ gate with angles
\begin{align}
\theta_{ij} &= \arccos(|a_{ij}|), \label{eq:theta2}\\
  \phi_{ij} &= -\alpha_{ij},
\label{eq:thetarho}
\end{align}
i.e.,
\begin{align}
R_z(2\phi_{ij}) & R_y(2\theta_{ij}) \ket0 \nonumber \\
& =
\begin{bmatrix}
e^{-\I \phi_{ij}} & \\
& e^{\I \phi_{ij}}
\end{bmatrix}
\begin{bmatrix}
\cos(\theta_{ij}) &           -\sin(\theta_{ij}) \\
\sin(\theta_{ij}) & \phantom{-}\cos(\theta_{ij})
\end{bmatrix}
\begin{bmatrix} 1 \\ 0 \end{bmatrix}, \nonumber \\
& =
\begin{bmatrix}
|a_{ij}| e^{\I \alpha_{ij}} \\ \sqrt{\smash[b]{1 - |a_{ij}|^2}} e^{-\I \alpha_{ij}}
\end{bmatrix}.
\end{align}

Our previous analysis extends to uniformly
controlled $R_z$ rotations because the conditions in~\eqref{eq:rotcond}
are satisfied for $R_z$ gates~\cite{Mottonen2004}.
It follows that we can implement the $O_A$ oracle for
complex-valued matrices as the product of uniformly
controlled $R_y$ and $R_z$ rotations
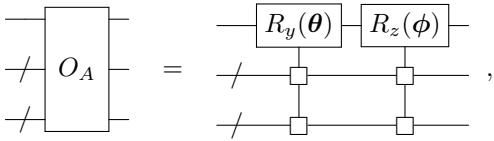
\[
\begin{myqcircuit*}{1}{0.75}
\lstick{} & \qw & \multigate{2}{O_A} &  \qw \\
\lstick{} & {/} \qw &        \ghost{O_A} &  \qw \\
\lstick{} & {/} \qw &        \ghost{O_A} &  \qw
\end{myqcircuit*}
\quad = \quad
\begin{myqcircuit}
& \qw & \gate{R_y(\bftheta)} & \gate{R_z(\bfphi)} &\qw \\
& {/} \qw & \ctrlsq{-1} & \ctrlsq{-1} & \qw \\
& {/} \qw & \ctrlsq{-1} & \ctrlsq{-1} & \qw
\end{myqcircuit}\ \ ,
\]
where the $R_y$ and $R_z$ rotations respectively set the magnitude
and phase of the matrix elements. The corresponding 
rotation angles can be computed separately through two independent
linear systems of the form~\eqref{eq:ls} by using the magnitude and phase
of the matrix data, respectively.
Hence, the gate complexity for $O_A$ oracles of complex-valued matrices
is twice the cost of real-valued matrices, while the asymptotic complexity
in terms of 1- and 2-qubit gates remains $\bigO(N^2)$.

The complete FABLE circuits for the real and complex case with the $O_A$ 
oracles implemented as uniformly controlled rotations are given in \cref{fig:FABLE-detail}.

\begin{figure}[hbtp]
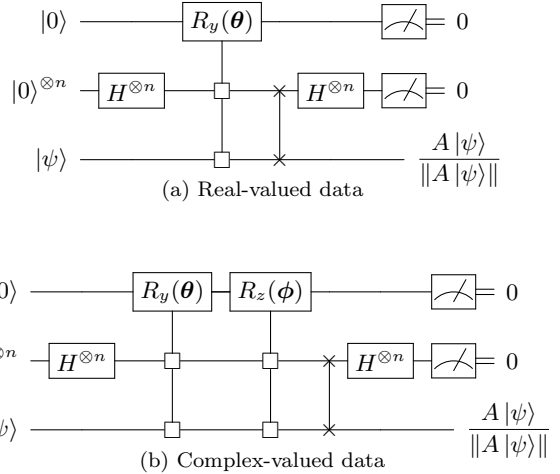

\centering
\subfloat[Real-valued data]{%
\(%
\begin{myqcircuit}[1.25]
\lstick{\ket{0}}             & \qw                  & \gate{R_y(\bftheta)} &
  \qw                        & \qw                  & \meter & \rstick{0}\cw \\
\lstick{\ket{0}^{\otimes n}} & \gate{H^{\otimes n}} &        \ctrlsq{-1} &
  \qswap                     & \gate{H^{\otimes n}} & \meter & \rstick{0}\cw \\
\lstick{\ket{\psi}}          & \qw                  &        \ctrlsq{-1} &
  \qswap\qwx                 & \qw                  &
  \rstick{\Frac{A\ket\psi}{\norm{A\ket\psi}}}\qw
\end{myqcircuit}%
\)%
}

\vspace{15pt}

\subfloat[Complex-valued data]{%
\(%
\begin{myqcircuit}[1.25]
\lstick{\ket{0}}             & \qw & \gate{R_y(\bftheta)} & \gate{R_z(\bfphi)}
  \qw                        & \qw & \qw            & \meter & \rstick{0}\cw \\
\lstick{\ket{0}^{\otimes n}} & \gate{H^{\otimes n}} & \ctrlsq{-1} & \ctrlsq{-1}&
  \qswap                     & \gate{H^{\otimes n}} & \meter & \rstick{0}\cw \\
\lstick{\ket{\psi}}          & \qw                  & \ctrlsq{-1} & \ctrlsq{-1}&
  \qswap\qwx                 & \qw                  &
  \rstick{\Frac{A\ket\psi}{\norm{A\ket\psi}}}\qw
\end{myqcircuit}%
\)%
}

\vspace{5pt}

\caption{FABLE quantum circuit structures for real and complex matrices
with $O_A$ oracles implemented as uniformly controlled rotations.\label{fig:FABLE-detail}}
\end{figure}

\section{Approximate Block-Encodings and Circuit Compression}
\label{sec:approx}

Thus far we have only considered exact implementations of $O_A$ resulting
in exact block-encodings. In this section, we focus the ``A'' in FABLE and
introduce \emph{approximate} block-encodings. We show how to compress
FABLE circuits and what the resulting approximation error of the
block-encoding is.

\subsection{Approximate block-encodings}
We begin with extending \Cref{def:BE} to approximate block-encodings
that only implement the target matrix up to a certain precision $\varepsilon$.

\begin{definition}
\label{def:ABE}
Let $a,n, m \inN$, $m = a + n$, and $\varepsilon\inR^+$.
Then an $m$-qubit unitary $U$ is an
${(\alpha, a, \varepsilon)}$-block-encoding of an $n$-qubit operator $A$ if
\begin{myequation}
\tilde A = \left(\bra{0}^{\otimes a} \otimes \eI_{n} \right) U
               \left(\ket{0}^{\otimes a} \otimes \eI_{n} \right),
\end{myequation}
and
\(
\normtwo[\big]{A - \alpha \tilde A} \leq \varepsilon.
\)
\end{definition}

The parameter $\varepsilon$ is the absolute \emph{error} on the block-encoding.
As before, we have that $\normtwo{\tilde A} \leq 1$ and
therefore $\normtwo{A} \leq \alpha + \varepsilon$.

\subsection{FABLE circuit compression and sparsification}

We illustrate the idea of the circuit compression algorithm for a uniformly controlled
rotation gate with 8 angles:
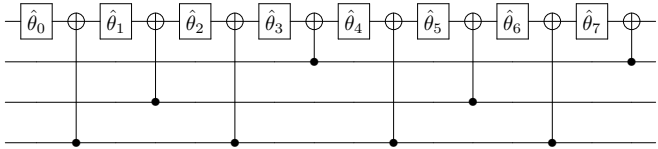
\[
\resizebox{\columnwidth}{!}{%
$
\begin{myqcircuit}
& \gate{\hat \theta_0}      & \targ    & \gate{\hat \theta_1} & \targ    & \gate{\hat \theta_2} & \targ    & \gate{\hat \theta_3} &  \targ    & \gate{\hat \theta_4}      & \targ    & \gate{\hat \theta_5} & \targ    & \gate{\hat \theta_6} & \targ    & \gate{\hat \theta_7} &  \targ    & \qw \\
& \qw                       & \qw      & \qw                  & \qw      & \qw                  & \qw      & \qw                  &  \ctrl{-1} & \qw                       & \qw      & \qw                  & \qw      & \qw                  & \qw      & \qw                  &  \ctrl{-1}      & \qw\\
& \qw                       & \qw      & \qw                  & \ctrl{-2} & \qw                  & \qw      & \qw                  &  \qw      & \qw                       & \qw      & \qw                  & \ctrl{-2} & \qw                  & \qw      & \qw                  &  \qw & \qw \\
& \qw                       & \ctrl{-3} & \qw                  & \qw      & \qw                  & \ctrl{-3} & \qw                  &  \qw      & \qw                       & \ctrl{-3} & \qw                  & \qw      & \qw                  & \ctrl{-3} & \qw                  &  \qw      & \qw \\
\end{myqcircuit}
$}
\]
For conciseness, we have omitted the labels from the rotation gates and only show their parameter.
The gates can be
$R_y$, $R_z$, or in general $R_{\alpha}$ gates with $\alpha$ a normalized linear combination of $\sigma_y$ and $\sigma_z$ because the conditions in~\eqref{eq:rotcond} are satisfied for all $R_{\alpha}$ gates~\cite{Mottonen2004}.

The compression algorithm uses a cutoff threshold $\delta_c \inR[+]$ and considers all $\hat\theta_i \leq \delta_c$ to be negligible.
Assume for the example above that $\hat\theta_2, \hat\theta_3, \hat\theta_4, \hat\theta_5, \hat\theta_6 \leq \delta_c$. This means that the respective
single-qubit rotations can be removed from the circuit, yielding
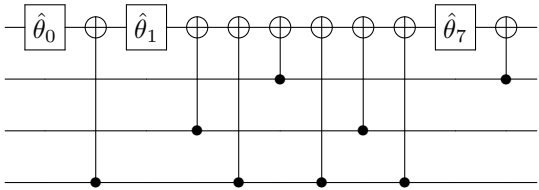
\[
\begin{myqcircuit}
& \gate{\hat \theta_0}      & \targ    & \gate{\hat \theta_1} & \targ    & \targ    &  \targ    & \targ    & \targ    & \targ    & \gate{\hat \theta_7} &  \targ    & \qw \\
& \qw                       & \qw      & \qw                  & \qw      & \qw      &  \ctrl{-1} & \qw      & \qw      & \qw      & \qw                  &  \ctrl{-1}      & \qw\\
& \qw                       & \qw      & \qw                  & \ctrl{-2} & \qw      &  \qw      & \qw      & \ctrl{-2} & \qw      & \qw                  &  \qw & \qw \\
& \qw                       & \ctrl{-3} & \qw                  & \qw      & \ctrl{-3} &  \qw      & \ctrl{-3} & \qw      & \ctrl{-3} & \qw                  &  \qw      & \qw \\
\end{myqcircuit}\ \ .
\]
The resulting circuit contains a series of consecutive CNOT gates that can be further simplified as they mutually commute and share the same target qubit.
Any two CNOT gates in a series of consecutive CNOT gates that have the same control qubit cancel out.
It follows that in the example circuit, 5 single qubit rotations and 4 CNOT gates can be removed,
yielding the compressed circuit:
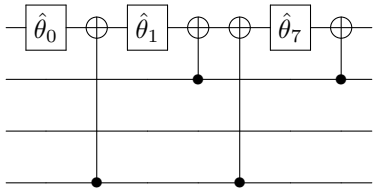
\[
\begin{myqcircuit}
& \gate{\hat \theta_0}      & \targ    & \gate{\hat \theta_1} & \targ    &  \targ    & \gate{\hat \theta_7} &  \targ    & \qw \\
& \qw                       & \qw      & \qw                  & \ctrl{-1} &  \qw      & \qw                  &  \ctrl{-1}      & \qw\\
& \qw                       & \qw      & \qw                  & \qw      &  \qw      & \qw                  &  \qw & \qw \\
& \qw                       & \ctrl{-3} & \qw                  & \qw      &  \ctrl{-3} & \qw                  &  \qw      & \qw \\
\end{myqcircuit}\ \ .
\]
%
The circuit compression algorithm consists of two steps:
\begin{enumerate}
\item[i.] Remove all rotation gates for angles $\hat\theta_i \leq \delta_c$ in $\hat\bftheta$ from
the circuit;
\item[ii.] Perform a parity check on the control qubits of the CNOT gates in each series of consecutive CNOT gates:
keep one CNOT gate with control qubit $i$ if there are an odd number of CNOT gates with control qubit $i$ in the series,
otherwise remove all CNOT gates with control qubit $i$.
\end{enumerate}

This procedure can be considered as data sparsification since it allows us to represent
the block-encoded matrix with fewer than $N^2$ parameters. However, since we perform the
sparsification on the $\hat\bftheta$ angles after the Walsh--Hadamard transform, a sparse
representation of $\hat\bftheta$ does not typically mean that $\bftheta$ and $A$ are sparse
in the usual sense of containing many zeroes. FABLE circuits are efficient for the class
of matrices that are sparse in the Walsh--Hadamard domain.

The following theorem relates the cutoff threshold $\delta_c$ used in the 
circuit compression algorithm to the first-order error on the block-encoding as defined in \Cref{def:ABE}.
We consider the case of real-valued data.
\begin{theorem}
\label{thm:bound}
For an $n$-qubit matrix $A\inR[N][N]$, $|a_{ij}| \leq 1$,
the FABLE circuit
with cutoff compression threshold $\delta_c \inR[+]$
gives an $(1/2^n, n+1, N^3\delta_c)$-block-encoding
of $A$ up to third order in $\delta_c$.
\end{theorem}
\begin{proof}
In order to prove that a cutoff compression threshold $\delta_c$ leads to an absolute
error of at most $N^3 \delta_c + \bigO(\delta_c^3)$ on the the block-encoding,
we start with the linear system \eqref{eq:ls} that relates the angles of the
uniformly controlled rotations $\hat\bftheta$ to the angles of the matrix query oracle $\bftheta$.

After thresholding the parameters $\hat \bftheta$ with cutoff $\delta_c$,
the uniformly controlled rotation is constructed with parameters
$\hat\bftheta + \bfdelta\hat\bftheta$,
where $|\delta\hat\theta_i| \leq \delta_c$.
It follows that $\| \bfdelta\hat\bftheta \|_2 \leq N \delta_c$.
This perturbs the angles in $O_A$ from $\bftheta$ to
\[
\tilde\bftheta = (\hat H^{\otimes 2n} P_G) (\hat\bftheta + \bfdelta\hat\bftheta).
\]
By linearity, the error on $O_A$ thus becomes 
\[
\bfdelta\bftheta = \tilde\bftheta - \bftheta = (\hat H^{\otimes 2n} P_G) \bfdelta\hat\bftheta,
\]
and we get that
\[
\normtwo{\bfdelta\bftheta} \leq \normtwo{\hat H^{\otimes 2n}}  \normtwo{P_G} \normtwo{\delta\hat\bftheta} \leq N^2 \delta_c,
\]
as $P_G$ is a unitary matrix and $\normtwo{\hat H^{\otimes 2n}} = N$.
This implies that the element-wise error is now only bounded by $|\delta \theta_i| = |\theta_i - \tilde\theta_i| \leq N^2 \delta_c$.
This relates to the element-wise error on $a_i$ as:
\begin{align*}
    |\delta a_i| = |a_i - \tilde a_i|
               & = |\cos(\theta_i) - \cos(\tilde\theta_i)| \\
               & = |\cos(\theta_i) - \cos(\theta_i + \delta\theta_i)| \\
               & = |2 \sin(\delta \theta_i / 2)\sin(\theta_i + (\delta\theta_i/2))|\\
               & \leq 2 |\sin(\delta \theta_i / 2)| \approx N^2 \delta_c + \bigO(\delta_{i}^3).
\end{align*}
In the final approximation, we used a truncated series expansion.
We thus have that $\|\bfdelta\bfa\|_2 \leq N^3\delta_c$.
As $\|A\|_2 \leq \|A\|_F = \|\vecc(A)\|_2$ we get the upper bound.
\end{proof}

Numerical results that verify this error bound are presented in 
\cref{fig:bound} which shows the result of noise-free QCLAB~\cite{qclab} simulations
for compressed FABLE circuits. The FABLE circuits are generated for randomly generated
matrices with entries drawn from the standard normal distribution, the matrices are 2 to 
7 qubit operators such that the FABLE circuits require 5 to 15 qubits.
We observe that the bound from \Cref{thm:bound} always holds but is overly pessimistic.
Not shown in \cref{fig:bound} are the majority of random realizations with an error close
to the $10^{-16}$ machine precision.

A similar analysis can be performed for the complex-valued case
where the error on the magnitude and phase of the matrix data have to be considered independently.

\begin{figure}[t]
\input{error_bound.tikz}
\caption{Error on simulated data and theoretical error bound from \Cref{thm:bound} for randomly
generated real-valued 2 ($4\times4$) to 7 ($128 \times 128$) qubit matrices for three compression cutoffs:
for $\delta_c = 10^{-6}$
the upper bound is the dashed line and the simulated data are the circles,
for $\delta_c = 10^{-4}$
the upper bound is the dotted line and the simulated data are the squares, and for $\delta_c = 10^{-2}$
the upper bound is the full line and the simulated data are the triangles.\label{fig:bound}}
\end{figure}
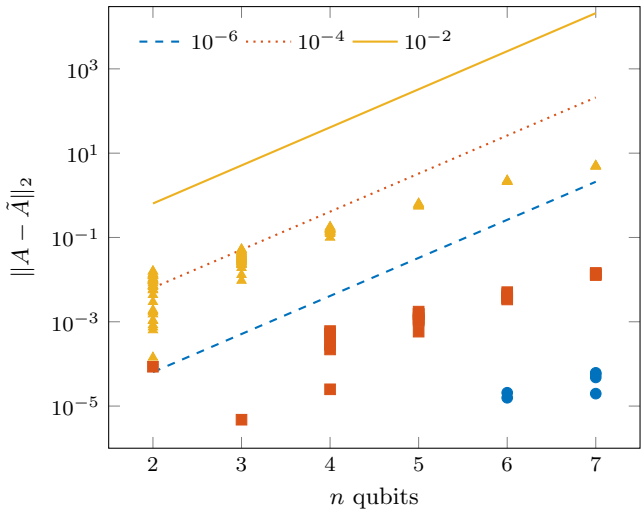

\section{Discussion and examples}\label{sec:ex}

FABLE circuits provide a fast and convenient way of generating quantum circuits consisting 
of simple 1- and 2-qubit gates that block-encode arbitrary matrices. Thanks to their
versatility, we expect them to become very useful to run and benchmark quantum algorithms
derived from the QSVT for small to medium-scale experiments in the NISQ era.

The conditions imposed on the block-encoded matrix are minimal: a FABLE circuit exists
for any matrix
that satisfies $|a_{ij}| \leq 1$ as this is the only requirement for real rotation angles
to exist according to \eqref{eq:theta} and \eqref{eq:theta2}.
For matrices of small norm, the probability of a successful measurement can vanish.
In the extreme case of the all-zero matrix, a FABLE encoding exists, but the probability
of measuring the desired state will also be zero.

For the remainder of this section, we show the gate complexities of FABLE circuits
for three different model problems: a Heisenberg Hamiltonian, a Fermi--Hubbard Hamiltonian in 1D and 2D,
and a discretized Laplacian operator in 1D and 2D.
We have selected these example problems as they perform well within FABLE encodings
and require relatively few gates. However, we do indicate that for more complicated model
problems this is not necessarily the case, which highlights the limitations of our approach.

\subsection{Heisenberg Hamiltonians}
Block-encodings of Hamiltonians are of particular interest as they can be used in the QSVT for Hamiltonian simulation~\cite{gisu2018} or for preparing ground and excited states~\cite{Lin2020}.
Where previous theoretical analysis showed that these methods have optimal asymptotic scaling in terms of
oracle query complexity, we can now directly use FABLE to obtain an upper bound for the asymptotic gate complexity for specific Hamiltonians.

We study the performance of FABLE for block-encoding localized Hamiltonians.
Specifically, we are interested in Heisenberg type spin chains Hamiltonians
\begin{align*}
H = &\sum_{i=1}^{n-1} J_x \, X\p{i}X\p{i+1} + J_y \,  Y\p{i}Y\p{i+1} + J_z \, Z\p{i}Z\p{i+1} \\ & + \sum_{i=1}^{n} h_z Z\p{i},
\end{align*}
where $X\p{i}$, $Y\p{i}$, and $Z\p{i}$ are the Pauli matrices,
\begin{equation}
X =
\begin{bmatrix}
0 & 1 \\ 1 & 0
\end{bmatrix}, \quad
Y =
\begin{bmatrix}
0 & -\I \\ \I & 0
\end{bmatrix}, \quad
Z =
\begin{bmatrix}
1 & 0 \\ 0 & -1
\end{bmatrix},
\end{equation}
acting on the $i$th qubit.

We consider a Heisenberg XXX model where $J_x = J_y = J_z$, $h_z = 0$ for systems of $2, \ldots, 7$ qubits.
We set the compression threshold $\delta_c$ to $\epsilon_m$, with $\epsilon_m$ the machine precision .
The CNOT and $R_y$ gate complexities are summarized in \cref{fig:HeisXXX_gates}.
We observe that even with such a small compression threshold, FABLE can give an accurate
encoding of a Heisenberg XXX model with just $50\%$ of the CNOT gates and $25\%$ of the $R_y$
rotations.

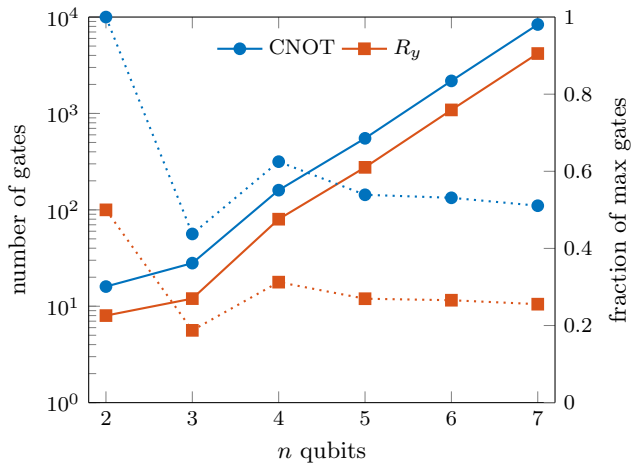
\begin{figure}[hbtp]
\input{HeisenbergXXX_gates.tikz}
\caption{CNOT (\emph{blue circles}) and $R_y$ (\emph{red squares}) gate complexity in function of $n$ for Heisenberg
XXX model without compression. The full lines show the absolute number of gates on the left $y$-axis and the dotted lines show
the fraction of the maximum number of gates ($4^n$) on the right $y$-axis.\label{fig:HeisXXX_gates}}
\end{figure}

More general Heisenberg models with different values for $J_x$, $J_y$ and $J_z$ or with an external
field ($h_z \neq 0$) can be encoded with FABLE, but their circuits cannot be compressed and sparsified
even though the corresponding matrices are sparse.

\subsection{Fermi--Hubbard model}

The second example we consider is a special case of the Fermi--Hubbard Hamiltonian,
\begin{equation}
H = - t \sum_{ij}\sum_{\sigma\in\lbrace \uparrow,\downarrow \rbrace} c^{\dagger}_{i\sigma}c_{j\sigma}
+ U \sum_i c^{\dagger}_{i\uparrow}c_{i\uparrow}c^{\dagger}_{i\downarrow}c_{i\downarrow},
\end{equation}
where $c^{\dagger}_{i\sigma}$ ($c_{i\sigma}$) is the creation (annihilation) operator for site $i$
and spin $\sigma$. The first term is the hopping term with strength $t$, the second term is the interaction term with strength $U$.
We generate the Hamiltonian through OpenFermion~\cite{openfermion} for the case $t=1$, $U=0$, i.e.,
non-interacting fermions on a 1D and 2D lattices with two spin orbitals per site and non-periodic boundary conditions.
The resulting Hamiltonians are mapped to qubits using the Bravyi--Kitaev transformation
and block-encoded with FABLE. The compression threshold is again set to $\epsilon_m$ and the
gate complexities of the FABLE circuits are listed in~\Cref{tab:hubbard}.
We observe that both in 1D and 2D the operators can be encoded with a relatively small
fraction of the maximum number of gates and that this fraction decreases for growing problem
size.
We ran the same experiment for interacting Hubbard Hamiltonians ($U \neq 0$) but for this
case, the matrices do not compress well and the maximum number of gates is required to
encode the Hamiltonians.

\begin{table*}[t]
\centering\small%
\begin{tabularx}{0.9\textwidth}{lcC|cC|cC}
\toprule
& & & \multicolumn{2}{c|}{CNOT} & \multicolumn{2}{c}{$R_y$} \\[5pt]
Model & Size & $n$ qubits & Gates & Fraction[\%] & Gates & Fraction[\%] \\
\midrule
\multirow{5}{*}{1D Hubbard} & 2 sites & 4 & 130 & 50.8 & 65 & 25.4 \\
                            & 3 sites & 6 & 1,098 & 26.8 & 513 & 12.5 \\
                            & 4 sites  & 8 & 6,666 & 10.2 & 3,073 & 4.7 \\
                            & 5 sites & 10 & 35,850 & 3.4 & 16,385 & 1.6\\
                            & 6 sites & 12 & 180,234 & 1.1 & 81,921 & 0.5\\[5pt]
\multirow{2}{*}{2D Hubbard} & $2 \times 2$ sites & 8 & 8,706 & 13.3 & 3,329 & 5.1\\
                            & $2 \times 3$ sites & 12 & 252,626 & 1.5 & 90,113 & 0.5\\
\bottomrule
\end{tabularx}
\caption{CNOT and $R_y$ gate complexities for 1D and 2D Hubbard models,
both in absolute number of gates and as a fraction of the maximum number
of gates ($4^n$).}
\label{tab:hubbard}
\end{table*}

\subsection{Elliptic partial differential equations}

As a second example we consider 1D and 2D discretized Laplace operators
which are relevant in the solution of elliptic partial differential equations
such as the Laplace equation $\Delta u = 0$ and the Poisson equation $\Delta u = f$.

We consider finite difference discretization of the second order derivatives. For example,
on a 1D equidistant grid with Dirichlet boundary conditions, we approximate the
second order derivate in point $x_j$ as:
\begin{equation}
u^{\prime\prime}_j \approx \Frac{u_{j+1} - 2u_j + u_{j-1}}{\Delta x^2},
\end{equation}
where $\Delta x$ is the step size.
This three-point stencil leads to the 1D discretized Laplace operator:
\begin{equation}
L_{xx} =
\begin{bmatrix}
2 & -1 & 0 & \cdots & * \\
-1 & 2 & -1 & \ddots & \vdots \\
0 & \ddots & \ddots & \ddots & 0 \\
\vdots & \ddots & -1 & 2 & -1 \\
* & \cdots & 0 & -1 & 2
\end{bmatrix},
\label{eq:laplace1D}
\end{equation}
where the entries $*$ in the lower-left and upper-right
corner are either both equal to $0$ for non-periodic boundary
conditions, or both equal to $-1$ for periodic boundary conditions.

In 2D, the discretized Laplace operator becomes the Kronecker sum
of discretizations along the $x$- and $y$-directions:
\begin{equation}
L 
= L_{xx} \oplus L_{yy}
= L_{xx} \otimes I + I \otimes L_{yy},
\label{eq:laplace2D}
\end{equation}
which corresponds to a five-point stencil.
This allows for periodic or non-periodic boundary conditions
and the number of discretization points can differ in both dimensions.
Second order derivatives are of great importance in quantum field
theory~\cite{QFT} and can be useful to simulate scalar fields with quantum
computers.

We generate 1D Laplacian matrices~\eqref{eq:laplace1D} on $2$ to $7$ qubits and
use FABLE to generate block-encoding circuits. The compression threshold is set
to $\epsilon_m$ such that an accurate block-encoding is obtained. The results are
summarized on the first row of \cref{fig:Laplace} for non-periodic and periodic
Dirichlet boundary conditions.
Similarly, the second row of \cref{fig:Laplace} shows the results of encoding 2D
Laplacians on different rectangular grids requiring at most $7$ qubits.

The following pertinent observations can be made from these results. First,
periodic boundary conditions lead to much reduced gate counts compared to
non-periodic boundary conditions. This is natural as there exist more structure
in the periodic Laplacians. Second, the 2D Laplacians can be compressed better
than 1D Laplacians and require in some cases fewer than 5\% of the maximum number
of gates. Third, for the 2D Laplacians, discretization grids with a smaller aspect
ratio, i.e., closer to a square, require in general fewer gates than more rectangular
grids for the same number of grid points.

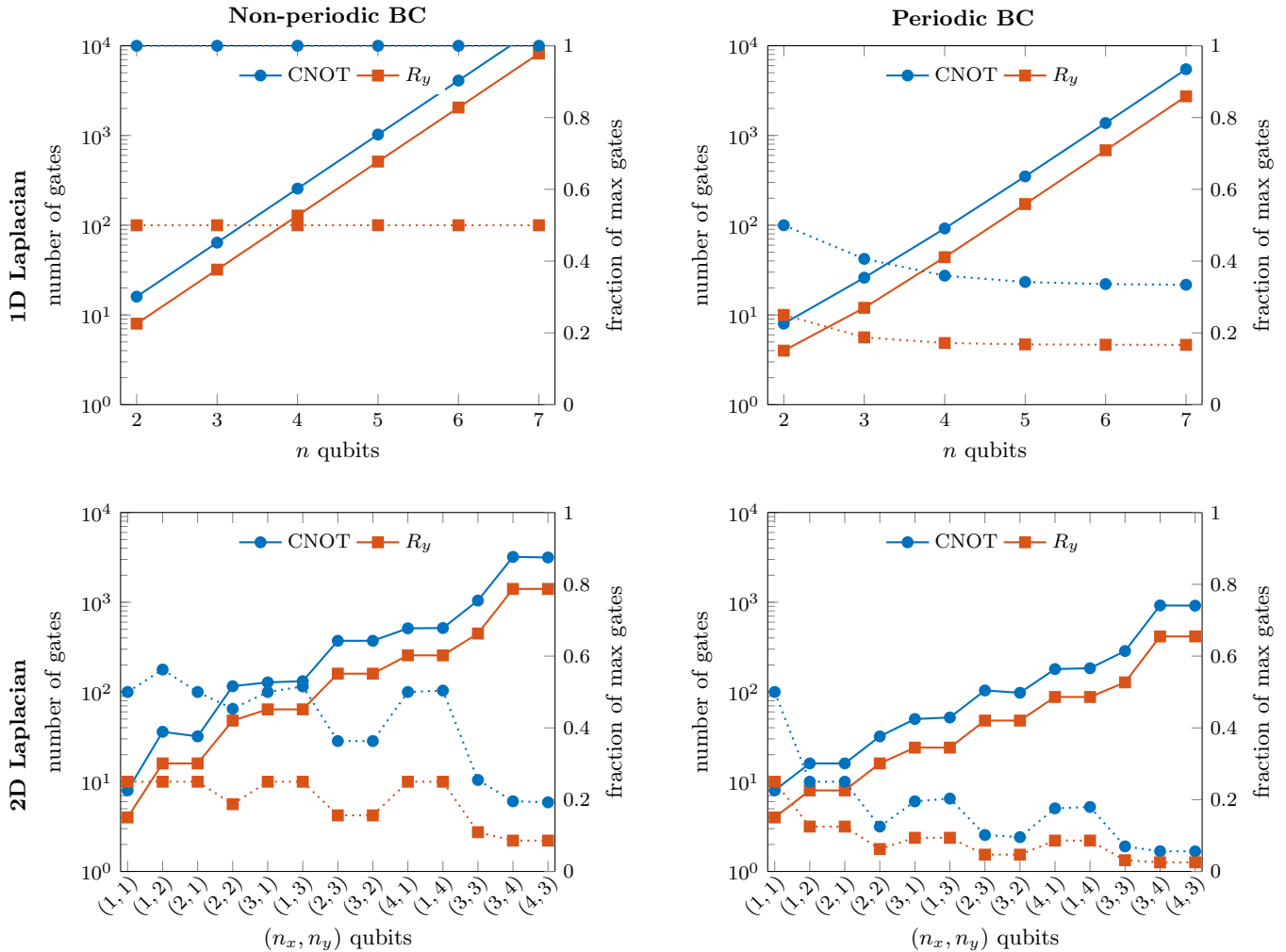
\begin{figure*}[t]
\centering
\begin{minipage}{0.5\textwidth}
\centering\small\textbf{Non-periodic BC}
\end{minipage}%
\hfill%
\begin{minipage}{0.5\textwidth}
\centering\small\textbf{Periodic BC}
\end{minipage}%

\rotatebox{90}{\qquad\qquad\qquad~\textbf{1D Laplacian}}
\input{laplacian_1d_no_pbc.tikz}
\hfill%
\input{laplacian_1d_pbc.tikz}

\vspace{10pt}

\rotatebox{90}{\qquad\qquad\qquad~\textbf{2D Laplacian}}
\input{laplacian_2d_no_pbc.tikz}%
\hfill
\input{laplacian_2d_pbc.tikz}

\caption{CNOT (\emph{blue circles}) and $R_y$ (\emph{red squares}) gate complexity for exact FABLE block-encodings of 1D and 2D Laplacian matrices with periodic and non-periodic boundary conditions.
The full lines show the absolute number of gates on the left $y$-axis and the dotted lines show
the fraction of the maximum number of gates ($4^n$) on the right $y$-axis.\label{fig:Laplace}}
\end{figure*}

\subsection{Quantum image encodings}
The circuit construction in FABLE is based on the concept of uniformly controlled rotations introduced in~\cite{Mottonen2004}.
This circuit construction method was recently used in the QPIXL framework
that unifies quantum image representations~\cite{amankwah2021}.
In all proposed quantum image representations and quantum image processing algorithms,
the image is encoded in a quantum state and the circuit implementation becomes a state
preparation problem.
FABLE, and block-encodings in general, can be directly used to encode image data
by embedding it in the unitary operator itself.
This alternative to previously proposed quantum image encodings could potentially
have benefits for certain quantum image processing tasks as it trivially preserves
the 2D structure in the image.

\section{Conclusion}\label{sec:conc}

In this paper, we have introduced FABLE, an algorithm for \emph{fast}
generation of quantum circuits that \emph{approximately} block-encode
arbitrary target matrices. 
FABLE circuits obtain the optimal asymptotic gate complexity for generic
dense operators and can be efficiently compressed for many problems of interest.
Circuit compression and sparsification can lead to a significant reduction
in gate complexity. More specifically, matrices that are sparse in the Walsh--Hadamard
domain can in general be efficiently encoded.
An interesting future research direction would be to precisely characterize the class of matrices
that are sparse in the Walsh--Hadamard domain as these have great potential for 
successful experimental realizations of quantum algorithms based on block-encodings.

We analyzed the relation between the compression threshold and the approximation error
in the block-encoding and provide an upper bound on this error in \Cref{thm:bound}
which shows a linear relation between the compression threshold, the problem size, and
the approximation error.
Our numerical simulations
show that this bound can be saturated for larger thresholds and problem sizes.

We illustrated FABLE on example problems ranging from Heisenberg and Hubbard Hamiltonians
to discretized Laplacian operators. These examples show that high compression levels are
feasible for certain structured problems, but we have observed that more general Hamiltonians
with more interaction terms do not compress well. This highlights the limitation of our direct
approach to encoding the matrix data.
Balancing direct FABLE encodings 
of smaller terms in the Hamiltonian expression with a
\emph{Linear Combination of Unitaries} (LCU)~\cite{Berry2015c} circuit construction
in order to combine the smaller terms into
a large-scale block-encoding can potentially mitigate this issue.
This is another future research direction.

\section*{Acknowledgment}

The authors would like to thank
Michael Kreschkuk and David Williams-Young
for their input and discussions which
have improved the quality of this work.

This work was supported by the Laboratory Directed Research and Development Program of Lawrence Berkeley National Laboratory and used resources of the National Energy Research Scientific Computing Center (NERSC), a U.S. Department of Energy Office of Science User Facility located at Lawrence Berkeley National Laboratory, operated under Contract No. DE-AC02-05CH11231.

\bibliographystyle{apsrev4-2}
\bibliography{references}

\end{document}

%% file: error_bound.tikz
\begin{tikzpicture}

\begin{axis}[%
width=\columnwidth,
xmin=1.5,
xmax=7.5,
xtick={2, 3, 4, 5, 6, 7},
xlabel={$n$ qubits},
ymode=log,
ymin=1e-6,
ymax=3e4,
yminorticks=true,
ylabel={$\|A - \tilde A\|_2$},
legend columns=3, 
legend style={at={(0.04,.97)}, anchor=north west, legend cell align=left, align=left, draw=none}
]
\addplot [color=myblue, thick, dashed]
  table[row sep=crcr]{%
2	6.4000e-05\\
3	5.1200e-04\\
4	4.0960e-03\\
5	3.2768e-02\\
6	2.6214e-01\\
7	2.0972e+00\\
};
\addlegendentry{$10^{-6}$}
\addplot [color=myblue, only marks, mark=*, mark options={solid, myblue}, forget plot]
  table[row sep=crcr]{%
6	1.57350298916867e-05\\
};
\addplot [color=myblue, only marks, mark=*, mark options={solid, myblue}, forget plot]
  table[row sep=crcr]{%
6	2.07230990861759e-05\\
};
\addplot [color=myblue, only marks, mark=*, mark options={solid, myblue}, forget plot]
  table[row sep=crcr]{%
7	1.97214242486127e-05\\
};
\addplot [color=myblue, only marks, mark=*, mark options={solid, myblue}, forget plot]
  table[row sep=crcr]{%
7	6.14055943697538e-05\\
};
\addplot [color=myblue, only marks, mark=*, mark options={solid, myblue}, forget plot]
  table[row sep=crcr]{%
7	4.77933364111731e-05\\
};
\addplot [color=myblue, only marks, mark=*, mark options={solid, myblue}, forget plot]
  table[row sep=crcr]{%
7	5.67571134766468e-05\\
};
\addplot [thick, color=myred, dotted]
  table[row sep=crcr]{%
2	6.4000e-03\\
3	5.1200e-02\\
4	4.0960e-01\\
5	3.2768e+00\\
6	2.6214e+01\\
7	2.0972e+02\\
};
\addlegendentry{$10^{-4}$}
\addplot [color=myred, only marks, mark=square*, mark options={solid, myred}, forget plot]
  table[row sep=crcr]{%
2	8.57722392958305e-05\\
};
\addplot [color=myred, only marks, mark=square*, mark options={solid, myred}, forget plot]
  table[row sep=crcr]{%
3	4.74575399286522e-06\\
};
\addplot [color=myred, only marks, mark=square*, mark options={solid, myred}, forget plot]
  table[row sep=crcr]{%
4	0.000529895656692309\\
};
\addplot [color=myred, only marks, mark=square*, mark options={solid, myred}, forget plot]
  table[row sep=crcr]{%
4	0.000220485035241949\\
};
\addplot [color=myred, only marks, mark=square*, mark options={solid, myred}, forget plot]
  table[row sep=crcr]{%
4	0.000601701217665961\\
};
\addplot [color=myred, only marks, mark=square*, mark options={solid, myred}, forget plot]
  table[row sep=crcr]{%
4	0.000283471805651636\\
};
\addplot [color=myred, only marks, mark=square*, mark options={solid, myred}, forget plot]
  table[row sep=crcr]{%
4	0.000287926481056317\\
};
\addplot [color=myred, only marks, mark=square*, mark options={solid, myred}, forget plot]
  table[row sep=crcr]{%
4	2.50008637872067e-05\\
};
\addplot [color=myred, only marks, mark=square*, mark options={solid, myred}, forget plot]
  table[row sep=crcr]{%
4	0.000517303031747349\\
};
\addplot [color=myred, only marks, mark=square*, mark options={solid, myred}, forget plot]
  table[row sep=crcr]{%
5	0.000583039427467217\\
};
\addplot [color=myred, only marks, mark=square*, mark options={solid, myred}, forget plot]
  table[row sep=crcr]{%
5	0.00174028810241351\\
};
\addplot [color=myred, only marks, mark=square*, mark options={solid, myred}, forget plot]
  table[row sep=crcr]{%
5	0.00154099676543763\\
};
\addplot [color=myred, only marks, mark=square*, mark options={solid, myred}, forget plot]
  table[row sep=crcr]{%
5	0.00105520573364732\\
};
\addplot [color=myred, only marks, mark=square*, mark options={solid, myred}, forget plot]
  table[row sep=crcr]{%
5	0.00151440020815414\\
};
\addplot [color=myred, only marks, mark=square*, mark options={solid, myred}, forget plot]
  table[row sep=crcr]{%
5	0.00110702489031296\\
};
\addplot [color=myred, only marks, mark=square*, mark options={solid, myred}, forget plot]
  table[row sep=crcr]{%
5	0.00112778024555959\\
};
\addplot [color=myred, only marks, mark=square*, mark options={solid, myred}, forget plot]
  table[row sep=crcr]{%
5	0.00123492128188006\\
};
\addplot [color=myred, only marks, mark=square*, mark options={solid, myred}, forget plot]
  table[row sep=crcr]{%
6	0.00450081233100212\\
};
\addplot [color=myred, only marks, mark=square*, mark options={solid, myred}, forget plot]
  table[row sep=crcr]{%
6	0.00360443594034237\\
};
\addplot [color=myred, only marks, mark=square*, mark options={solid, myred}, forget plot]
  table[row sep=crcr]{%
6	0.00504441107830099\\
};
\addplot [color=myred, only marks, mark=square*, mark options={solid, myred}, forget plot]
  table[row sep=crcr]{%
6	0.00338111671666035\\
};
\addplot [color=myred, only marks, mark=square*, mark options={solid, myred}, forget plot]
  table[row sep=crcr]{%
6	0.00366692460400199\\
};
\addplot [color=myred, only marks, mark=square*, mark options={solid, myred}, forget plot]
  table[row sep=crcr]{%
6	0.00399577766411586\\
};
\addplot [color=myred, only marks, mark=square*, mark options={solid, myred}, forget plot]
  table[row sep=crcr]{%
7	0.0138864438433539\\
};
\addplot [color=myred, only marks, mark=square*, mark options={solid, myred}, forget plot]
  table[row sep=crcr]{%
7	0.0128425180731325\\
};
\addplot [color=myred, only marks, mark=square*, mark options={solid, myred}, forget plot]
  table[row sep=crcr]{%
7	0.0143859943535984\\
};
\addplot [color=myred, only marks, mark=square*, mark options={solid, myred}, forget plot]
  table[row sep=crcr]{%
7	0.0131830119691417\\
};
\addplot [thick, color=myorange]
  table[row sep=crcr]{%
2	6.4000e-01\\
3	5.1200e+00\\
4	4.0960e+01\\
5	3.2768e+02\\
6	2.6214e+03\\
7	2.0972e+04\\
};
\addlegendentry{$10^{-2}$}

\addplot [color=myorange, only marks, mark=triangle*, mark options={solid, myorange}, forget plot]
  table[row sep=crcr]{%
2	0.001690317986816\\
};
\addplot [color=myorange, only marks, mark=triangle*, mark options={solid, myorange}, forget plot]
  table[row sep=crcr]{%
2	0.00709329433380881\\
};
\addplot [color=myorange, only marks, mark=triangle*, mark options={solid, myorange}, forget plot]
  table[row sep=crcr]{%
2	0.00905328174584298\\
};
\addplot [color=myorange, only marks, mark=triangle*, mark options={solid, myorange}, forget plot]
  table[row sep=crcr]{%
2	0.00105409611282077\\
};
\addplot [color=myorange, only marks, mark=triangle*, mark options={solid, myorange}, forget plot]
  table[row sep=crcr]{%
2	0.000142611212010458\\
};
\addplot [color=myorange, only marks, mark=triangle*, mark options={solid, myorange}, forget plot]
  table[row sep=crcr]{%
2	0.0124767322034161\\
};
\addplot [color=myorange, only marks, mark=triangle*, mark options={solid, myorange}, forget plot]
  table[row sep=crcr]{%
2	0.00958256816632792\\
};
\addplot [color=myorange, only marks, mark=triangle*, mark options={solid, myorange}, forget plot]
  table[row sep=crcr]{%
2	0.000635061694587452\\
};
\addplot [color=myorange, only marks, mark=triangle*, mark options={solid, myorange}, forget plot]
  table[row sep=crcr]{%
2	0.012997047617601\\
};
\addplot [color=myorange, only marks, mark=triangle*, mark options={solid, myorange}, forget plot]
  table[row sep=crcr]{%
2	0.015862741275439\\
};
\addplot [color=myorange, only marks, mark=triangle*, mark options={solid, myorange}, forget plot]
  table[row sep=crcr]{%
2	0.00107735282092544\\
};
\addplot [color=myorange, only marks, mark=triangle*, mark options={solid, myorange}, forget plot]
  table[row sep=crcr]{%
2	0.0121303989356253\\
};
\addplot [color=myorange, only marks, mark=triangle*, mark options={solid, myorange}, forget plot]
  table[row sep=crcr]{%
2	0.010639193508724\\
};
\addplot [color=myorange, only marks, mark=triangle*, mark options={solid, myorange}, forget plot]
  table[row sep=crcr]{%
2	0.00590578474599414\\
};
\addplot [color=myorange, only marks, mark=triangle*, mark options={solid, myorange}, forget plot]
  table[row sep=crcr]{%
2	0.0137669107065215\\
};
\addplot [color=myorange, only marks, mark=triangle*, mark options={solid, myorange}, forget plot]
  table[row sep=crcr]{%
2	0.0107160960761616\\
};
\addplot [color=myorange, only marks, mark=triangle*, mark options={solid, myorange}, forget plot]
  table[row sep=crcr]{%
2	0.00990064872995101\\
};
\addplot [color=myorange, only marks, mark=triangle*, mark options={solid, myorange}, forget plot]
  table[row sep=crcr]{%
2	0.00678087484669892\\
};
\addplot [color=myorange, only marks, mark=triangle*, mark options={solid, myorange}, forget plot]
  table[row sep=crcr]{%
2	0.0156697828204332\\
};
\addplot [color=myorange, only marks, mark=triangle*, mark options={solid, myorange}, forget plot]
  table[row sep=crcr]{%
2	0.0008400928757283\\
};
\addplot [color=myorange, only marks, mark=triangle*, mark options={solid, myorange}, forget plot]
  table[row sep=crcr]{%
2	0.00984429913859316\\
};
\addplot [color=myorange, only marks, mark=triangle*, mark options={solid, myorange}, forget plot]
  table[row sep=crcr]{%
2	0.0133349104498645\\
};
\addplot [color=myorange, only marks, mark=triangle*, mark options={solid, myorange}, forget plot]
  table[row sep=crcr]{%
2	0.0125118712452154\\
};
\addplot [color=myorange, only marks, mark=triangle*, mark options={solid, myorange}, forget plot]
  table[row sep=crcr]{%
2	0.00151054732288452\\
};
\addplot [color=myorange, only marks, mark=triangle*, mark options={solid, myorange}, forget plot]
  table[row sep=crcr]{%
2	0.00735120402627145\\
};
\addplot [color=myorange, only marks, mark=triangle*, mark options={solid, myorange}, forget plot]
  table[row sep=crcr]{%
2	0.00946861048431381\\
};
\addplot [color=myorange, only marks, mark=triangle*, mark options={solid, myorange}, forget plot]
  table[row sep=crcr]{%
2	0.00074178090462466\\
};
\addplot [color=myorange, only marks, mark=triangle*, mark options={solid, myorange}, forget plot]
  table[row sep=crcr]{%
2	0.00296839473770127\\
};
\addplot [color=myorange, only marks, mark=triangle*, mark options={solid, myorange}, forget plot]
  table[row sep=crcr]{%
2	0.00825823110369364\\
};
\addplot [color=myorange, only marks, mark=triangle*, mark options={solid, myorange}, forget plot]
  table[row sep=crcr]{%
2	0.00563621378028095\\
};
\addplot [color=myorange, only marks, mark=triangle*, mark options={solid, myorange}, forget plot]
  table[row sep=crcr]{%
2	0.00180991037275049\\
};
\addplot [color=myorange, only marks, mark=triangle*, mark options={solid, myorange}, forget plot]
  table[row sep=crcr]{%
2	0.00189514392162285\\
};
\addplot [color=myorange, only marks, mark=triangle*, mark options={solid, myorange}, forget plot]
  table[row sep=crcr]{%
2	0.00436039376193953\\
};
\addplot [color=myorange, only marks, mark=triangle*, mark options={solid, myorange}, forget plot]
  table[row sep=crcr]{%
2	0.00708766622434768\\
};
\addplot [color=myorange, only marks, mark=triangle*, mark options={solid, myorange}, forget plot]
  table[row sep=crcr]{%
2	0.00968280417096776\\
};
\addplot [color=myorange, only marks, mark=triangle*, mark options={solid, myorange}, forget plot]
  table[row sep=crcr]{%
3	0.035073593915384\\
};
\addplot [color=myorange, only marks, mark=triangle*, mark options={solid, myorange}, forget plot]
  table[row sep=crcr]{%
3	0.0352913528841026\\
};
\addplot [color=myorange, only marks, mark=triangle*, mark options={solid, myorange}, forget plot]
  table[row sep=crcr]{%
3	0.0358976837298081\\
};
\addplot [color=myorange, only marks, mark=triangle*, mark options={solid, myorange}, forget plot]
  table[row sep=crcr]{%
3	0.0295757361388844\\
};
\addplot [color=myorange, only marks, mark=triangle*, mark options={solid, myorange}, forget plot]
  table[row sep=crcr]{%
3	0.0294985612492617\\
};
\addplot [color=myorange, only marks, mark=triangle*, mark options={solid, myorange}, forget plot]
  table[row sep=crcr]{%
3	0.032517278245363\\
};
\addplot [color=myorange, only marks, mark=triangle*, mark options={solid, myorange}, forget plot]
  table[row sep=crcr]{%
3	0.02389125275711\\
};
\addplot [color=myorange, only marks, mark=triangle*, mark options={solid, myorange}, forget plot]
  table[row sep=crcr]{%
3	0.0454742109253262\\
};
\addplot [color=myorange, only marks, mark=triangle*, mark options={solid, myorange}, forget plot]
  table[row sep=crcr]{%
3	0.0337000113671048\\
};
\addplot [color=myorange, only marks, mark=triangle*, mark options={solid, myorange}, forget plot]
  table[row sep=crcr]{%
3	0.0384652437242937\\
};
\addplot [color=myorange, only marks, mark=triangle*, mark options={solid, myorange}, forget plot]
  table[row sep=crcr]{%
3	0.0295138835158128\\
};
\addplot [color=myorange, only marks, mark=triangle*, mark options={solid, myorange}, forget plot]
  table[row sep=crcr]{%
3	0.0276246748305451\\
};
\addplot [color=myorange, only marks, mark=triangle*, mark options={solid, myorange}, forget plot]
  table[row sep=crcr]{%
3	0.0223629021589254\\
};
\addplot [color=myorange, only marks, mark=triangle*, mark options={solid, myorange}, forget plot]
  table[row sep=crcr]{%
3	0.0094791235177852\\
};
\addplot [color=myorange, only marks, mark=triangle*, mark options={solid, myorange}, forget plot]
  table[row sep=crcr]{%
3	0.0408155927884728\\
};
\addplot [color=myorange, only marks, mark=triangle*, mark options={solid, myorange}, forget plot]
  table[row sep=crcr]{%
3	0.0314810876066672\\
};
\addplot [color=myorange, only marks, mark=triangle*, mark options={solid, myorange}, forget plot]
  table[row sep=crcr]{%
3	0.0336273479352863\\
};
\addplot [color=myorange, only marks, mark=triangle*, mark options={solid, myorange}, forget plot]
  table[row sep=crcr]{%
3	0.0523035796310598\\
};
\addplot [color=myorange, only marks, mark=triangle*, mark options={solid, myorange}, forget plot]
  table[row sep=crcr]{%
3	0.0320227718883929\\
};
\addplot [color=myorange, only marks, mark=triangle*, mark options={solid, myorange}, forget plot]
  table[row sep=crcr]{%
3	0.0130336207981053\\
};
\addplot [color=myorange, only marks, mark=triangle*, mark options={solid, myorange}, forget plot]
  table[row sep=crcr]{%
3	0.0346244453607978\\
};
\addplot [color=myorange, only marks, mark=triangle*, mark options={solid, myorange}, forget plot]
  table[row sep=crcr]{%
3	0.0344929488871629\\
};
\addplot [color=myorange, only marks, mark=triangle*, mark options={solid, myorange}, forget plot]
  table[row sep=crcr]{%
3	0.0397979678424632\\
};
\addplot [color=myorange, only marks, mark=triangle*, mark options={solid, myorange}, forget plot]
  table[row sep=crcr]{%
3	0.0266386421663862\\
};
\addplot [color=myorange, only marks, mark=triangle*, mark options={solid, myorange}, forget plot]
  table[row sep=crcr]{%
3	0.0373474993992792\\
};
\addplot [color=myorange, only marks, mark=triangle*, mark options={solid, myorange}, forget plot]
  table[row sep=crcr]{%
3	0.0470139511547279\\
};
\addplot [color=myorange, only marks, mark=triangle*, mark options={solid, myorange}, forget plot]
  table[row sep=crcr]{%
3	0.0453060150294419\\
};
\addplot [color=myorange, only marks, mark=triangle*, mark options={solid, myorange}, forget plot]
  table[row sep=crcr]{%
3	0.031271020491008\\
};
\addplot [color=myorange, only marks, mark=triangle*, mark options={solid, myorange}, forget plot]
  table[row sep=crcr]{%
3	0.0357203531319012\\
};
\addplot [color=myorange, only marks, mark=triangle*, mark options={solid, myorange}, forget plot]
  table[row sep=crcr]{%
3	0.0295362395738225\\
};
\addplot [color=myorange, only marks, mark=triangle*, mark options={solid, myorange}, forget plot]
  table[row sep=crcr]{%
3	0.0281526706432288\\
};
\addplot [color=myorange, only marks, mark=triangle*, mark options={solid, myorange}, forget plot]
  table[row sep=crcr]{%
3	0.0442639529667903\\
};
\addplot [color=myorange, only marks, mark=triangle*, mark options={solid, myorange}, forget plot]
  table[row sep=crcr]{%
3	0.0264543982222889\\
};
\addplot [color=myorange, only marks, mark=triangle*, mark options={solid, myorange}, forget plot]
  table[row sep=crcr]{%
3	0.0190869831260213\\
};
\addplot [color=myorange, only marks, mark=triangle*, mark options={solid, myorange}, forget plot]
  table[row sep=crcr]{%
3	0.0346643398145059\\
};
\addplot [color=myorange, only marks, mark=triangle*, mark options={solid, myorange}, forget plot]
  table[row sep=crcr]{%
3	0.0226676354758641\\
};
\addplot [color=myorange, only marks, mark=triangle*, mark options={solid, myorange}, forget plot]
  table[row sep=crcr]{%
3	0.0423661148170613\\
};
\addplot [color=myorange, only marks, mark=triangle*, mark options={solid, myorange}, forget plot]
  table[row sep=crcr]{%
3	0.0337953404806096\\
};
\addplot [color=myorange, only marks, mark=triangle*, mark options={solid, myorange}, forget plot]
  table[row sep=crcr]{%
3	0.0248523044602829\\
};
\addplot [color=myorange, only marks, mark=triangle*, mark options={solid, myorange}, forget plot]
  table[row sep=crcr]{%
3	0.052083784725452\\
};
\addplot [color=myorange, only marks, mark=triangle*, mark options={solid, myorange}, forget plot]
  table[row sep=crcr]{%
3	0.0502832865472084\\
};
\addplot [color=myorange, only marks, mark=triangle*, mark options={solid, myorange}, forget plot]
  table[row sep=crcr]{%
3	0.0237572161199292\\
};
\addplot [color=myorange, only marks, mark=triangle*, mark options={solid, myorange}, forget plot]
  table[row sep=crcr]{%
3	0.0411947491360382\\
};
\addplot [color=myorange, only marks, mark=triangle*, mark options={solid, myorange}, forget plot]
  table[row sep=crcr]{%
3	0.0367727771457619\\
};
\addplot [color=myorange, only marks, mark=triangle*, mark options={solid, myorange}, forget plot]
  table[row sep=crcr]{%
3	0.036466432253205\\
};
\addplot [color=myorange, only marks, mark=triangle*, mark options={solid, myorange}, forget plot]
  table[row sep=crcr]{%
3	0.0404447308496133\\
};
\addplot [color=myorange, only marks, mark=triangle*, mark options={solid, myorange}, forget plot]
  table[row sep=crcr]{%
3	0.0349990532299889\\
};
\addplot [color=myorange, only marks, mark=triangle*, mark options={solid, myorange}, forget plot]
  table[row sep=crcr]{%
3	0.034131378431852\\
};
\addplot [color=myorange, only marks, mark=triangle*, mark options={solid, myorange}, forget plot]
  table[row sep=crcr]{%
3	0.0330353709291261\\
};
\addplot [color=myorange, only marks, mark=triangle*, mark options={solid, myorange}, forget plot]
  table[row sep=crcr]{%
3	0.0494683743524772\\
};
\addplot [color=myorange, only marks, mark=triangle*, mark options={solid, myorange}, forget plot]
  table[row sep=crcr]{%
4	0.141314048551637\\
};
\addplot [color=myorange, only marks, mark=triangle*, mark options={solid, myorange}, forget plot]
  table[row sep=crcr]{%
4	0.15371663462901\\
};
\addplot [color=myorange, only marks, mark=triangle*, mark options={solid, myorange}, forget plot]
  table[row sep=crcr]{%
4	0.161430574492133\\
};
\addplot [color=myorange, only marks, mark=triangle*, mark options={solid, myorange}, forget plot]
  table[row sep=crcr]{%
4	0.0998362665174892\\
};
\addplot [color=myorange, only marks, mark=triangle*, mark options={solid, myorange}, forget plot]
  table[row sep=crcr]{%
4	0.138243469420901\\
};
\addplot [color=myorange, only marks, mark=triangle*, mark options={solid, myorange}, forget plot]
  table[row sep=crcr]{%
4	0.130851710947646\\
};
\addplot [color=myorange, only marks, mark=triangle*, mark options={solid, myorange}, forget plot]
  table[row sep=crcr]{%
4	0.135742989370279\\
};
\addplot [color=myorange, only marks, mark=triangle*, mark options={solid, myorange}, forget plot]
  table[row sep=crcr]{%
4	0.169112049161485\\
};
\addplot [color=myorange, only marks, mark=triangle*, mark options={solid, myorange}, forget plot]
  table[row sep=crcr]{%
4	0.175629483121836\\
};
\addplot [color=myorange, only marks, mark=triangle*, mark options={solid, myorange}, forget plot]
  table[row sep=crcr]{%
4	0.159217379249655\\
};
\addplot [color=myorange, only marks, mark=triangle*, mark options={solid, myorange}, forget plot]
  table[row sep=crcr]{%
4	0.156470387509559\\
};
\addplot [color=myorange, only marks, mark=triangle*, mark options={solid, myorange}, forget plot]
  table[row sep=crcr]{%
4	0.134496897847258\\
};
\addplot [color=myorange, only marks, mark=triangle*, mark options={solid, myorange}, forget plot]
  table[row sep=crcr]{%
4	0.148306745543066\\
};
\addplot [color=myorange, only marks, mark=triangle*, mark options={solid, myorange}, forget plot]
  table[row sep=crcr]{%
4	0.165538620456095\\
};
\addplot [color=myorange, only marks, mark=triangle*, mark options={solid, myorange}, forget plot]
  table[row sep=crcr]{%
4	0.166782918724189\\
};
\addplot [color=myorange, only marks, mark=triangle*, mark options={solid, myorange}, forget plot]
  table[row sep=crcr]{%
4	0.150672506023581\\
};
\addplot [color=myorange, only marks, mark=triangle*, mark options={solid, myorange}, forget plot]
  table[row sep=crcr]{%
4	0.178156975167509\\
};
\addplot [color=myorange, only marks, mark=triangle*, mark options={solid, myorange}, forget plot]
  table[row sep=crcr]{%
4	0.161601419732624\\
};
\addplot [color=myorange, only marks, mark=triangle*, mark options={solid, myorange}, forget plot]
  table[row sep=crcr]{%
4	0.155840074816622\\
};
\addplot [color=myorange, only marks, mark=triangle*, mark options={solid, myorange}, forget plot]
  table[row sep=crcr]{%
4	0.121973820252417\\
};
\addplot [color=myorange, only marks, mark=triangle*, mark options={solid, myorange}, forget plot]
  table[row sep=crcr]{%
5	0.593490927904581\\
};
\addplot [color=myorange, only marks, mark=triangle*, mark options={solid, myorange}, forget plot]
  table[row sep=crcr]{%
5	0.589048655617598\\
};
\addplot [color=myorange, only marks, mark=triangle*, mark options={solid, myorange}, forget plot]
  table[row sep=crcr]{%
5	0.601854828771966\\
};
\addplot [color=myorange, only marks, mark=triangle*, mark options={solid, myorange}, forget plot]
  table[row sep=crcr]{%
5	0.640320807880537\\
};
\addplot [color=myorange, only marks, mark=triangle*, mark options={solid, myorange}, forget plot]
  table[row sep=crcr]{%
5	0.552946014320235\\
};
\addplot [color=myorange, only marks, mark=triangle*, mark options={solid, myorange}, forget plot]
  table[row sep=crcr]{%
5	0.558670679289857\\
};
\addplot [color=myorange, only marks, mark=triangle*, mark options={solid, myorange}, forget plot]
  table[row sep=crcr]{%
5	0.632254227678372\\
};
\addplot [color=myorange, only marks, mark=triangle*, mark options={solid, myorange}, forget plot]
  table[row sep=crcr]{%
5	0.55183818205205\\
};
\addplot [color=myorange, only marks, mark=triangle*, mark options={solid, myorange}, forget plot]
  table[row sep=crcr]{%
6	2.26557124082943\\
};
\addplot [color=myorange, only marks, mark=triangle*, mark options={solid, myorange}, forget plot]
  table[row sep=crcr]{%
6	2.13045773825611\\
};
\addplot [color=myorange, only marks, mark=triangle*, mark options={solid, myorange}, forget plot]
  table[row sep=crcr]{%
6	2.10315565856507\\
};
\addplot [color=myorange, only marks, mark=triangle*, mark options={solid, myorange}, forget plot]
  table[row sep=crcr]{%
6	2.18414540190346\\
};
\addplot [color=myorange, only marks, mark=triangle*, mark options={solid, myorange}, forget plot]
  table[row sep=crcr]{%
6	2.13680356944332\\
};
\addplot [color=myorange, only marks, mark=triangle*, mark options={solid, myorange}, forget plot]
  table[row sep=crcr]{%
6	2.10114357846958\\
};
\addplot [color=myorange, only marks, mark=triangle*, mark options={solid, myorange}, forget plot]
  table[row sep=crcr]{%
7	4.96787511523334\\
};
\addplot [color=myorange, only marks, mark=triangle*, mark options={solid, myorange}, forget plot]
  table[row sep=crcr]{%
7	4.92972969694265\\
};
\addplot [color=myorange, only marks, mark=triangle*, mark options={solid, myorange}, forget plot]
  table[row sep=crcr]{%
7	4.82073188209234\\
};
\addplot [color=myorange, only marks, mark=triangle*, mark options={solid, myorange}, forget plot]
  table[row sep=crcr]{%
7	5.08037091382985\\
};

\end{axis}
\end{tikzpicture}%

%% file: HeisenbergXXX_gates.tikz
\begin{tikzpicture}

  
\begin{axis}[%
width=0.9\columnwidth,
axis y line*=left,
ylabel={number of gates},
ymode=log,
xmin=1.8,
xmax=7.2,
ymin=1,
ymax=1e4,
xtick={2, 3, 4, 5, 6, 7},
xlabel={$n$ qubits},
legend columns=2, 
legend style={at={(0.5,.97)}, anchor=north, legend cell align=left, align=left, draw=none}
]
\addplot [color=myblue, thick, mark=*, mark options={solid, myblue, fill=myblue}]
  table[row sep=crcr]{%
2	16\\
3	28\\
4	160\\
5	552\\
6	2176\\
7	8368\\
};
\addlegendentry{CNOT}

\addplot [color=myred, thick, mark=square*, mark options={solid, myred, fill=myred}]
  table[row sep=crcr]{%
2	8\\
3	12\\
4	80\\
5	276\\
6	1088\\
7	4184\\
};
\addlegendentry{$R_y$}
\end{axis}

\begin{axis}[%
width=0.9\columnwidth,
axis y line*=right,
axis x line=none,
xmin=1.8,
xmax=7.2,
ylabel={\small fraction of max gates},
ymin=0,
ymax=1,
ylabel near ticks,
]
\addplot [color=myblue, thick, dotted, mark=*, mark options={solid, myblue, fill=myblue}, forget plot]
  table[row sep=crcr]{%
2	1\\
3	0.4375\\
4	0.625\\
5	0.5390625\\
6	0.53125\\
7	0.5107421875\\
};
\addplot [color=myred, thick, dotted, mark=square*, mark options={solid, myred, fill=myred}, forget plot]
  table[row sep=crcr]{%
2	0.5\\
3	0.1875\\
4	0.3125\\
5	0.26953125\\
6	0.265625\\
7	0.25537109375\\
};
\end{axis}

\end{tikzpicture}%

%% file: laplacian_1d_no_pbc.tikz
\begin{tikzpicture}

  
\begin{axis}[%
width=0.43\linewidth,
axis y line*=left,
ylabel={number of gates},
ymode=log,
xmin=1.8,
xmax=7.2,
ymin=1,
ymax=1e4,
xtick={2, 3, 4, 5, 6, 7},
xlabel={$n$ qubits},
legend columns=2, 
legend style={at={(0.5,.97)}, anchor=north, legend cell align=left, align=left, draw=none},
]
\addplot [color=myblue, thick, mark=*, mark options={solid, myblue, fill=myblue}]
  table[row sep=crcr]{%
2	16\\
3	64\\
4	256\\
5	1024\\
6	4096\\
7	16384\\
};
\addlegendentry{CNOT}

\addplot [color=myred, thick, mark=square*, mark options={solid, myred, fill=myred}]
  table[row sep=crcr]{%
2	8\\
3	32\\
4	128\\
5	512\\
6	2048\\
7	8192\\
};
\addlegendentry{$R_y$}
\end{axis}

\begin{axis}[%
width=0.43\linewidth,
axis y line*=right,
axis x line=none,
xmin=1.8,
xmax=7.2,
ylabel={\small fraction of max gates},
ymin=0,
ymax=1,
ylabel near ticks,
]
\addplot [color=myblue, thick, dotted, mark=*, mark options={solid, myblue, fill=myblue}, forget plot]
  table[row sep=crcr]{%
2	1\\
3	1\\
4	1\\
5	1\\
6	1\\
7	1\\
};
\addplot [color=myred, thick, dotted, mark=square*, mark options={solid, myred, fill=myred}, forget plot]
  table[row sep=crcr]{%
2	0.5\\
3	0.5\\
4	0.5\\
5	0.5\\
6	0.5\\
7	0.5\\
};
\end{axis}

\end{tikzpicture}%

%% file: laplacian_1d_pbc.tikz
\begin{tikzpicture}

  
\begin{axis}[%
width=0.43\linewidth,
axis y line*=left,
ylabel={number of gates},
ymode=log,
xmin=1.8,
xmax=7.2,
ymin=1,
ymax=1e4,
xtick={2, 3, 4, 5, 6, 7},
xlabel={$n$ qubits},
legend columns=2, 
legend style={at={(0.5,.97)}, anchor=north, legend cell align=left, align=left, draw=none}
]
\addplot [color=myblue, thick, mark=*, mark options={solid, myblue, fill=myblue}]
  table[row sep=crcr]{%
2	8\\
3	26\\
4	92\\
5	350\\
6	1376\\
7	5474\\
};
\addlegendentry{CNOT}

\addplot [color=myred, thick, mark=square*, mark options={solid, myred, fill=myred}]
  table[row sep=crcr]{%
2	4\\
3	12\\
4	44\\
5	172\\
6	684\\
7	2732\\
};
\addlegendentry{$R_y$}
\end{axis}

\begin{axis}[%
width=0.43\linewidth,
axis y line*=right,
axis x line=none,
xmin=1.8,
xmax=7.2,
ylabel={\small fraction of max gates},
ymin=0,
ymax=1,
ylabel near ticks,
]
\addplot [color=myblue, thick, dotted, mark=*, mark options={solid, myblue, fill=myblue}, forget plot]
  table[row sep=crcr]{%
2	0.5\\
3	0.40625\\
4	0.359375\\
5	0.341796875\\
6	0.3359375\\
7	0.3341064453125\\
};
\addplot [color=myred, thick, dotted, mark=square*, mark options={solid, myred, fill=myred}, forget plot]
  table[row sep=crcr]{%
2	0.25\\
3	0.1875\\
4	0.171875\\
5	0.16796875\\
6	0.1669921875\\
7	0.166748046875\\
};
\end{axis}

\end{tikzpicture}%

%% file: laplacian_2d_no_pbc.tikz
\begin{tikzpicture}

  
\begin{axis}[%
width=0.43\linewidth,
axis y line*=left,
ylabel={number of gates},
ymode=log,
xmin=0.8,
xmax=13.2,
ymin=1,
ymax=1e4,
xtick={1, 2, 3, 4, 5, 6, 7, 8, 9, 10, 11, 12, 13},
xticklabels={{$(1, 1)$}, 
             {$(1, 2)$}, 
             {$(2, 1)$}, 
             {$(2, 2)$}, 
             {$(3, 1)$}, 
             {$(1, 3)$},
             {$(2, 3)$},
             {$(3, 2)$},
             {$(4, 1)$},
             {$(1, 4)$},
             {$(3, 3)$},
             {$(3, 4)$},
             {$(4, 3)$},
             },
x tick label style={rotate=45, anchor=north east, inner sep=0mm},
xlabel={$(n_x, n_y)$ qubits},
legend columns=2, 
legend style={at={(0.5,.97)}, anchor=north, legend cell align=left, align=left, draw=none}
]
\addplot [color=myblue, thick, mark=*, mark options={solid, myblue, fill=myblue}]
  table[row sep=crcr]{%
1	8\\
2	36\\
3	32\\
4	116\\
5	128\\
6	132\\
7	372\\
8	372\\
9	512\\
10	516\\
11	1044\\
12	3204\\
13	3156\\
};
\addlegendentry{CNOT}

\addplot [color=myred, thick, mark=square*, mark options={solid, myred, fill=myred}]
  table[row sep=crcr]{%
1	4\\
2	16\\
3	16\\
4	48\\
5	64\\
6	64\\
7	160\\
8	160\\
9	256\\
10	256\\
11	448\\
12	1408\\
13	1408\\
};
\addlegendentry{$R_y$}
\end{axis}

\begin{axis}[%
width=0.43\linewidth,
axis y line*=right,
axis x line=none,
ylabel={\small fraction of max gates},
ymin=0,
ymax=1,
xmin=0.8,
xmax=13.2,
ylabel near ticks,
]
\addplot [color=myblue, thick, dotted, mark=*, mark options={solid, myblue, fill=myblue}, forget plot]
  table[row sep=crcr]{%
1	0.5\\
2	0.5625\\
3	0.5\\
4	0.453125\\
5	0.5\\
6	0.515625\\
7	0.36328125\\
8	0.36328125\\
9	0.5\\
10	0.50390625\\
11	0.2548828125\\
12	0.195556640625\\
13	0.192626953125\\
};
\addplot [color=myred, thick, dotted, mark=square*, mark options={solid, myred, fill=myred}, forget plot]
  table[row sep=crcr]{%
1	0.25\\
2	0.25\\
3	0.25\\
4	0.1875\\
5	0.25\\
6	0.25\\
7	0.15625\\
8	0.15625\\
9	0.25\\
10	0.25\\
11	0.109375\\
12	0.0859375\\
13	0.0859375\\
};
\end{axis}

\end{tikzpicture}%

%% file: laplacian_2d_pbc.tikz
\begin{tikzpicture}

  
\begin{axis}[%
width=0.43\linewidth,
axis y line*=left,
ylabel={number of gates},
ymode=log,
xmin=0.8,
xmax=13.2,
ymin=1,
ymax=1e4,
xtick={1, 2, 3, 4, 5, 6, 7, 8, 9, 10, 11, 12, 13},
xticklabels={{$(1, 1)$}, 
             {$(1, 2)$}, 
             {$(2, 1)$}, 
             {$(2, 2)$}, 
             {$(3, 1)$}, 
             {$(1, 3)$},
             {$(2, 3)$},
             {$(3, 2)$},
             {$(4, 1)$},
             {$(1, 4)$},
             {$(3, 3)$},
             {$(3, 4)$},
             {$(4, 3)$},
             },
x tick label style={rotate=45, anchor=north east, inner sep=0mm},
xlabel={$(n_x, n_y)$ qubits},
legend columns=2, 
legend style={at={(0.5,.97)}, anchor=north, legend cell align=left, align=left, draw=none}
]
\addplot [color=myblue, thick, mark=*, mark options={solid, myblue, fill=myblue}]
  table[row sep=crcr]{%
1	8\\
2	16\\
3	16\\
4	32\\
5	50\\
6	52\\
7	104\\
8	98\\
9	180\\
10	184\\
11	286\\
12	922\\
13	916\\
};
\addlegendentry{CNOT}

\addplot [color=myred, thick, mark=square*, mark options={solid, myred, fill=myred}]
  table[row sep=crcr]{%
1	4\\
2	8\\
3	8\\
4	16\\
5	24\\
6	24\\
7	48\\
8	48\\
9	88\\
10	88\\
11	128\\
12	416\\
13	416\\
};
\addlegendentry{$R_y$}
\end{axis}

\begin{axis}[%
width=0.43\linewidth,
axis y line*=right,
axis x line=none,
ylabel={\small fraction of max gates},
xmin=0.8,
xmax=13.2,
ymin=0,
ymax=1,
ylabel near ticks,
]
\addplot [color=myblue, thick, dotted, mark=*, mark options={solid, myblue, fill=myblue}, forget plot]
  table[row sep=crcr]{%
1	0.5\\
2	0.25\\
3	0.25\\
4	0.125\\
5	0.1953125\\
6	0.203125\\
7	0.1015625\\
8	0.095703125\\
9	0.17578125\\
10	0.1796875\\
11	0.06982421875\\
12	0.0562744140625\\
13	0.055908203125\\
};
\addplot [color=myred, thick, dotted, mark=square*, mark options={solid, myred, fill=myred}, forget plot]
  table[row sep=crcr]{%
1	0.25\\
2	0.125\\
3	0.125\\
4	0.0625\\
5	0.09375\\
6	0.09375\\
7	0.046875\\
8	0.046875\\
9	0.0859375\\
10	0.0859375\\
11	0.03125\\
12	0.025390625\\
13	0.025390625\\
};
\end{axis}

\end{tikzpicture}%